\pdfoutput=1
\documentclass[pdftex,pra,aps,onecolumn,twoside,nofootinbib]{revtex4-2}

\usepackage{dsfont}
\usepackage[T1]{fontenc}
\usepackage[utf8]{inputenc}
\usepackage{float}
\usepackage{braket}
\usepackage{listings}
\usepackage{indentfirst}
\usepackage{enumerate}
\usepackage{url}
\usepackage[dvipsnames]{xcolor}
\usepackage[stable]{footmisc}
\usepackage{fancyhdr}
\usepackage[all]{xy}

\usepackage{verbatim}
\usepackage{bbm}
\usepackage{etoolbox}
\usepackage{graphicx}
\usepackage{youngtab}
\usepackage{mathtools}
\usepackage{enumerate}
\usepackage{textcomp}
\usepackage{newcent}
\usepackage[sc,noBBpl]{mathpazo}
\usepackage{amsmath,amsfonts,amssymb,amsthm}
\usepackage[mathscr]{eucal}
\usepackage{tcolorbox}
\usepackage{ragged2e}
\usepackage{bbold}
\usepackage{hyperref}
\hypersetup{colorlinks,allcolors=black}

\theoremstyle{plain}
\newtheorem{theorem}{Theorem}
\newtheorem{lemma}[theorem]{Lemma}

\newtheorem{definition}[theorem]{Definition}

\newtheorem{observation}[theorem]{Observation}
\newtheoremstyle{note}{\topsep}{\topsep}{\slshape}{}{\scshape}{}{ }{}
\theoremstyle{note}

\newcommand\tr{\operatorname{Tr}}

\newcommand\B{\widetilde{B}}
\newcommand{\<}{\langle}
\renewcommand{\>}{\rangle}

\newcommand\be{\begin{equation}}
\newcommand\ee{\end{equation}}
\newcommand\bea{\begin{array}}
	\newcommand\eea{\end{array}}
\newcommand\ben{\begin{eqnarray}}
\newcommand\een{\end{eqnarray}}
\newcommand\ot{\otimes}

\newcommand\bei{\begin{itemize}}
	\newcommand\eei{\end{itemize}}
\newcommand\bee{\begin{enumerate}}
	\newcommand\eee{\end{enumerate}}
\newcommand{\A}{\mathcal{A}_{n}'(d)}
\newcommand{\Hom}{\operatorname{Hom}}
\newcommand{\Span}{\operatorname{span}}

\begin{document} 	
\title{Minimal Port-based Teleportation}
\author{Sergii Strelchuk$^1$, Micha{\l} Studzi\'nski$^{2,}\footnote{corresponding author, email: studzinski.m.g@gmail.com}$}
	\affiliation{$^1$DAMTP, Centre for Mathematical Sciences, University of Cambridge, Cambridge CB30WA, UK\\
	$^2$Institute of Theoretical Physics and Astrophysics, University of Gda\'nsk, National Quantum Information Centre, 80-952 Gda\'nsk, Poland}
\begin{abstract}
There are two types of port-based teleportation (PBT) protocols: deterministic -- when the state always arrives to the receiver but is imperfectly transmitted and probabilistic -- when the state reaches the receiver intact with high probability. We introduce the minimal set of requirements that define a feasible PBT protocol and construct a simple PBT protocol that satisfies these requirements: it teleports an unknown state of a qubit with success probability $p_{succ}=1-\frac{N+2}{2^{N+1}}$ and fidelity $1-O(\frac{1}{N})$ with the resource state consisting of $N$ maximally entangled states. This protocol is not reducible from either the deterministic or probabilistic PBT protocol. We define the corresponding efficient superdense coding protocols which transmit more classical bits with fewer maximally entangled states. Furthermore, we introduce rigorous methods for comparing and converting between different PBT protocols.

\end{abstract}
\maketitle

\section{Introduction}

Quantum teleportation, introduced at the dawn of quantum information processing era, remains one of the pillars of quantum information transfer~\cite{bennett_teleporting_1993}. The concept of sending an unknown state of the quantum system without physically transferring the medium found its use in an impressive range of applications~\cite{pirandola2015advances}.  In recent years, Ishizaka and Hiroshima introduced new teleportation protocols with properties that were previously believed to be unattainable. These protocols go under a name of port-based teleportation (PBT)~\cite{ishizaka2008asymptotic} and they possess a counter-intuitive property that appears to be at odds with non-signalling principle of quantum mechanics, namely, that the teleported state requires no correction and is readily available for use after the sender performs a measurement and sends classical communication. This sparked a series of studies that investigate fundamental capabilities of different PBT protocols and their corresponding resource requirements such as the type of measurements and the amount of entanglement needed for the protocol to function~\cite{strelchuk_generalized_2013,stud2020A,Studzinski_2022,mozrzymas2021optimal,MozJPA,christandl2021asymptotic}. 

PBT protocols found wide-ranging applications in cryptography and instantaneous non-local computation~\cite{beigi2011simplified}, they were instrumental in establishing a link between interaction complexity and entanglement in non-local computation and holography~\cite{may2022complexity}, they established a fundamental link between quantum communication complexity advantage and a
violation of a Bell inequality~\cite{buhrman_quantum_2016}, fundamental limitations for quantum channels discrimination by designing adaptive protocols called PBT stretching~\cite{pirandola2019fundamental} and others~\cite{pereira2021characterising,quintino2021quantum,PhysRevLett.122.170502}.

Finding optimal protocols that work with quantum states of arbitrary dimensions necessitates the use of representation theory. The complexity of mathematical formalism precludes us from developing satisfying physical intuition about PBT protocols -- in stark contrast to the elegance of the first teleportation protocol introduced in 1993~\cite{bennett_teleporting_1993}. Moreover, the latter requires maximally-entangled states to operate most efficiently, whereas there exist flavours of port-based teleportation protocols which achieve the best-in-class performance using entangled states which are very different from maximally entangled states. 

The underlying motivation for this work comes from the seminal work of R. Werner~\cite{werner2001all} who considered a set of remarkably distinct problems which can be reduced to performing (the first 1993-style) quantum teleportation. The emergence of distinctly different teleportation protocols raised a number of fascinating questions about the fundamental building blocks of quantum information: what other forms of transferring quantum information from one subsystem to another exist? PBT protocols offer a distinct alternative to the first teleportation scheme. We thus were motivated by a series of questions: Are there many other distinct  protocols which result in quantum state transfer akin to the ordinary teleportation and PBT-like protocols? Can one treat the latter as a single class of protocols or, perhaps, there is a number of fundamentally different protocols within PBT? By now it has become clear that one can define numerous of PBT protocols, each different to another in that there was no straightforward way to provide a black-box reduction between them. Our work aims to address the above by providing two conceptual insights. 
First, we endeavour to derive PBT from `minimal' requirements. In the spirit of seminal research direction where one constructs quantum theories from fewest possible axioms~\cite{hardy2001quantum, masanes2011derivation}, we apply the same reasoning to the construction of teleportation protocols: what is the minimal number of requirements that yield a viable construction of the port-based teleportation? We introduce the first such `minimal' set of requirements that yield viable teleportation protocols together with the associated superdense coding schemes. Formalising the requirements immediately led to a construction of a new `lean’ or so-called minimal protocol. 
What's noteworthy, is that this protocol provably does not reduce to interpolation between the existing protocols and which has the exponentially improved  scaling of probability of success. Second, given two PBT protocols is there a systematic way of interconverting them without having to consider the fine-grained details of the implementation? An affirmative answer to this question would provide a black-box transformation of trade-offs between imperfect but guaranteed and perfect but probabilistic information transfer LOCC protocols, which is interesting in its own right. 

%This motivates one to look for: a) the simplest port-based teleportation protocols (in a certain rigorous sense) b) methods that allow for a consistent comparison of port-based protocols. The former would enable us to look for a minimal set of requirements that would yield a port-based teleportation protocol. These requirements manifest themselves in form of natural constraints on measurements and resource states. 

In our work, address both of these points: first, we outline the minimal set of assumptions that yield a working port-based teleportation protocol. We introduce a new protocol that cannot be reduced to any known protocols and that satisfies this set of minimal requirements. This protocol is exponentially more efficient than any known probabilistic PBT (pPBT), hence it cannot be obtained from any such protocol. At the same time, it uses special `denoised' measurement operators that provably cannot result in deterministic PBT (dPBT). 

Second, we develop methods for comparing port-based teleportation protocols in terms of their resource states which enable one to estimate their entanglement per port and distinctness. We further introduce methods for the conversion between the protocols and determine the conditions when it is possible to turn one type of PBT into another. We conclude by introducing a new family of superdense coding protocols which send more classical information using less entanglement compared to previously known protocol~\cite{ishizaka_remarks_2015}. 

In Section~\ref{prelim}, we provide self-contained operational and mathematical preliminaries. This is followed by a minimal set of assumptions that a PBT protocol must satisfy in Section~\ref{minimalrecs} together with a new, simplified PBT protocol. Remarkably, this remarkable simplicity does not come at a cost of reduced performance. To give a flavour of our results, in Table~\ref{table:entFPBT0} we collect known results on the performance of various qubits PBT schemes compared with the minimal PBT.
\begin{center}
\begin{table}[h!]
	\begin{tabular}{c|c|c}

Teleportation protocol & Entanglement fidelity $F$ & Average success probability $p_{succ}$\\
\hline
	Non-optimised deterministic PBT & $F=1-\mathcal{O}(1/N) $ & 1\\[0.1cm]

	Optimised deterministic PBT & $F=1-\mathcal{O}(1/N^2) $ & 1\\[0.1cm]

\hline
        Non-optimised probabilistic PBT & 1 & $p_{succ}=1-\mathcal{O}(1/\sqrt{N}) $ \\[0.1cm]

	Optimised probabilistic PBT & 1 & $p_{succ}=1-\mathcal{O}(1/N) $ \\[0.1cm]

 \hline
       {\bf Minimal PBT} & $F=1-\mathcal{O}(1/N) $ & $p_{succ}=1-\frac{N+2}{2^{N+1}}$

	\end{tabular}
\caption{Asymptotic behaviour of all known variants of PBT compared with introduced in this manuscript minimal PBT -- the qubit case. The minimal PBT offers exponentially better scaling in $N$ compared to optimised probabilistic PBT for average probability success, even with the non-optimised resource state included in this table. For entanglement fidelity, minimal PBT offers the same scaling with the number of ports $N$. In Section~\ref{intrpolated_prot} we present a more detailed discussion on the efficiency of mPBT protocol.}
	\label{table:entFPBT0}
\end{table}
\end{center} 

In Section~\ref{converting} we show how to convert different types of PBT protocols. The non-trivial regime that needs to be rigorously investigated is transforming pPBT into dPBT.

In Section~\ref{sec:fids} we study PBT protocols from a different viewpoint, by investigating the properties of their respective resource states. We introduce partial ordering on the PBT resource states by  considering the fidelity between resource states. We find that even in the case of two similarly performing protocols, their underlying resource states are drastically different. This ordering enables us to find the most `frugal' teleportation protocol: the one which achieves the highest performance with substantially lower entanglement requirements. 

Lastly, in Section~\ref{superdense} we turn to superdense coding protocols, which are dual to any teleportation protocol. While ordinary superdense coding protocols are well-understood in the context of original teleportation, very little is known about their dual-PBT versions, with only one known example in~\cite{ishizaka_remarks_2015}. We show how to take an arbitrary dPBT protocol with an established lower bound on fidelity and compute the corresponding performance of the superdense coding protocol. In particular, we find that there exist superdense coding protocols that are capable of transmitting the same amount of classical information as in~\cite{ishizaka_remarks_2015}, but using significantly less entanglement. The appendix contains the necessary source code to calculate all the quantities throughout the paper.

\section{Preliminaries}\label{prelim}
\subsection{Operational preliminaries of PBT}
\label{prelim:A}
The aim of any teleportation protocol is to transfer a given unknown quantum state $\psi_C$ from one party (Alice) to another (Bob). This is achieved by first distributing a resource state $\rho_{AB}$, the sender holding subsystem labelled $A$ and the receiver holding $B$. Then Alice performs a joint measurement from a set $\{M_i\}_i$ on $\psi_C\otimes\rho_{AB}$ and records the outcome $i\in {\mathds{N}}_0$. The latter is then communicated to Bob using a classical channel. After Bob receives $i$, he performs one of two actions (depending on the particulars of the protocol): he either discards all the subsystems except for the one that is identified by $i$ or he aborts the protocol (if $i$ encodes the error flag, for example when $i=0$). In the first case, the teleported state of $\psi_C$ is located in the subsystem identified by $i$.

In what follows we present a short summary of the PBT schemes. All the results contained here have been proven elsewhere and for more details, we encourage readers to see the following papers where the formalism of PBT has been developed~\cite{ishizaka2008asymptotic,ishizaka_quantum_2009,Studzinski2017,StuNJP,Studzinski_2022}.

We start with introducing resources used by parties wishing to apply the PBT scheme. Namely, parties share $N$ copies of $d$-dimensional maximally entangled states, each of them called \textit{port}:
\be
\label{resource}
|\Psi\>_{AB}=(O_A \otimes \mathbf{1}_B)|\Psi^+\>_{AB}=(O_A \otimes \mathbf{1}_B)|\psi^+\>_{A_1B_1}\otimes |\psi^+\>_{A_2B_2}\otimes \cdots \otimes |\psi^+\>_{A_NB_N} \ ,
\ee
where $A=A_1A_2\cdots A_N$, $B=B_1B_2\cdots B_N$, and $O_A$, with normalisation constraint $\tr(O_A^{\dagger}O_A)=d^N$, is a global operation applied by Alice to increase the efficiency of the protocol~\cite{ishizaka_quantum_2009,Studzinski2017,StuNJP}. In the case when  $O_A=\mathbf{1}_A$, we deal with $N$ maximally entangled state and call the total state $|\Psi\>_{AB}=|\Psi^+\>_{AB}$ as a  \textit{non-optimised resource state} and the whole protocol is called \textit{non-optimal PBT}. For optimal schemes,  the explicit forms of $O_A$ are presented in Appendix~\ref{summary} for the reader's convenience, as well as in cited above literature. Further  we refer to the state $|\Psi\>_{AB}=(O_A \otimes \mathbf{1}_B)|\Psi^+\>_{AB}$ as a \textit{optimised resource state} and the whole protocol is called \textit{optimal PBT}. Counterintuitively, due to the fact that $O_A\neq \mathbf{1}_A$ and $O_A$ is not unitary, the resource state in optimal schemes yields superior performance while not being maximally entangled. %Further, we also use the following simplifying notation $\Psi_{AB}=|\Psi\>\<\Psi|_{AB}$ and   $\Psi_{AB}^+=|\Psi^+\>\<\Psi^+|_{AB}$.

%We can distinguish two variants of the protocol: 
Presently, there are two flavours of the PBT protocols: %dPBT or pPBT. In the former case, the state always arrives to the receiver, but may be distorted, and in the latter, the state is not guaranteed to be transferred to Bob. However, when it arrives, there is no distortion. 
\begin{itemize}
    \item \textit{Deterministic protocol (dPBT):} An unknown quantum state $\psi_{C}$ is always transmitted to the receiver but the transmission is imperfect.
    The teleportation channel $\mathcal{N}_{C\rightarrow \B}$  is of the following form:
\be
\label{ch1}
\begin{split}
	\mathcal{N}_{C\rightarrow \B}\left(\psi_{C} \right)=\sum_{i=1}^N \tr_{AC}\left[\widetilde{\Pi}_{i}^{AC} \left(\left(O_A\ot \mathbf{1}_{\B} \right)\sigma^{A\B}_i \left(O_A^{\dagger}\ot \mathbf{1}_{\B} \right) \ot \psi_{C}\right)\right],
\end{split}
\ee
 where by $\tr_{AC}$ denotes partial trace over all systems $AC$  but $\B$.  The states $\sigma_{A_i\B}$ are called \textit{signal states}:
\be
\label{eq:signals}
\sigma^{A\B}_i:= \frac{1}{d^{N-1}}\mathbf{1}_{\overline{A}_i}\otimes P^+_{A_i\B} \ ,
\ee
where $P^+_{A_i\B}$ is projector on maximally entangled state between systems $A_i$ and $\B$. 
To assess the quality of the protocol, one can evaluate entanglement fidelity $F(\mathcal{N}_{C\rightarrow \B})$ of teleportation channel $\mathcal{N}_{C\rightarrow \B}$ when teleporting a subsystem $C$ of maximally entangled state $P^+_{CD}$, and computing overlap with the state after perfect transmission $P^+_{\B D}$:
\be
\label{F_det}
F(\mathcal{N}_{C\rightarrow \B})=\tr\left[P^+_{\B D}(\mathcal{N}_{C\rightarrow \B}\otimes \mathbb{1}_D)(P^+_{CD})\right]=\frac{1}{d^2}\sum_{i=1}^N\tr\left[\left(O_A^{\dagger}\otimes \mathbf{1}_{\B}\right)\widetilde{\Pi}_i^{A\B}\left(O_A\otimes \mathbf{1}_{\B}\right)\sigma^{A\B}_i\right],
\ee
where $\mathbb{1}_D$ denotes identity channel leaving system $D$ untouched.
For an arbitrary dimension $d$ the fidelity $F(\mathcal{N}_{C\rightarrow \B})$ has been evaluated explicitly using methods coming from group representation theory~\cite{Studzinski2017,StuNJP,majenz2}.
Due to the recent result presented in \cite{leditzky2020optimality}, we know that measurements in the form of \textit{square-root measurements} (SRM) are optimal in both PBT versions (non- and optimal PBT). The optimal measurements in the non-optimal case are:
\begin{equation}
	\label{eq:measurements}
	\forall 1\leq i\leq N \qquad \Pi_i^{AC}=\frac{1}{\sqrt{\rho}}\sigma_{A_iC}\frac{1}{\sqrt{\rho}},\quad \text{where}\quad \rho=\sum_{i=1}^N\sigma_{A_iC}.
\end{equation}
The operator $\rho^{-1}$ is restricted to the support of $\rho$, so to ensure summation of all POVMs to identity $\mathbf{1}_{AC}$ on the whole space $(\mathbb{C}^d)^{\otimes N+1}$, we add to every $\Pi_i^{AC}$ an excess term $\Delta/N$, where $\Delta=\mathbf{1}_{AC}-\sum_{i=1}^N\Pi_i^{AC}$. However, this extra term does not change the entanglement fidelity $F(\mathcal{N}_{C\rightarrow \B})$ of the channel $\mathcal{N}_{C\rightarrow \B}$, see also~\cite{ishizaka2008asymptotic,ishizaka_quantum_2009,Studzinski2017} for the detailed explanation.
    
    \item \textit{Probabilistic protocol (pPBT):} From~\cite{ishizaka_quantum_2009} we know that in the  probabilistic scheme, the process of teleportation sometimes fails, however when it succeeds the state is faithfully teleported to Bob with fidelity $F(\mathcal{N}_{C\rightarrow \B})=1$. In this scheme Alice has access to $N+1$ POVMs $\{M_0^{AC},M_1^{AC},\ldots,M_N^{AC}\}$ with measurement $M_0^{AC}$ corresponding to the failure of teleportation procedure. An additional POVM $M_0^{AC}$ makes the teleportation channel $\mathcal{N}_{C\rightarrow \B}$ trace non-preserving. The efficiency of the protocol is described by the average probability of success $p_{succ}$ of the scheme equals to~\cite{ishizaka_quantum_2009,Studzinski2017}:
\begin{equation}
\label{psucca}
p_{succ}=\frac{1}{d^{N+1}}\sum_{i=1}^N\tr\left[\widetilde{M}_i^{AC}\right],
\end{equation}
where $\widetilde{M}_i^{AC}=O_A^{\dagger}M_i^{AC}O_A$. In the same manner, we define a probability of failure $p_{fail}$ corresponding to POVM $M_0^{AC}$ averaged over all input states which of course satisfies relation $p_{fail}=1-p_{succ}$.
The requirement of perfect transmission $F(\mathcal{N})=1$ constraints form of allowed measurements accessible for Alice, namely can be of the following form~\cite{ishizaka_quantum_2009,Studzinski2017}:
\begin{equation}
\label{ort}
\forall 1\leq i\leq N \qquad \widetilde{M}_i^{AC}=P^+_{A_iC}\otimes \Theta_{\overline{A}_i}.
\end{equation}
where $\overline{A}_i$ denotes all states $A_1A_2\cdots A_N$ but $A_i$.
The optimal form of the operators $\Theta_{\overline{A}_i}$ in both versions of the pPBT protocols (optimal and non-optimal) was computed for qubits and qudits in~\cite{ishizaka_quantum_2009,Studzinski2017}.
\end{itemize}
In both cases (deterministic and probabilistic) we used $O_A$ to denote the optimising operation. However, when the  operators differ we will further write $\widetilde{O}_A$ for  dPBT and $O_A$ for pPBT. Such a distinction will be very helpful in Section~\ref{sec:fids} where we calculate overlaps between the resource states. In the case of pPBT protocol we also use the notion of non-optimised and optimised resource state as it was done above for dPBT scheme.

\subsection{Mathematical preliminaries of PBT}
\label{math_intro}

We will now review basic representation-theoretic preliminaries. We briefly describe here the irreducible representation formalism for the permutation group $S(n)$ with its algebra $\mathbb{C}[S(n)]$~\cite{FultonSchur,Schur-Weyl} and collect the main results regarding the algebra of partially transposed permutation operators $\A$~\cite{MozJPA,Studzinski_2022}. 
%For brevity, we omit the proofs from the representation theory, referring the reader to the relevant sources.

%We make extensive use of the symmetric group $S(n)$ and the algebra $\A$ of partially transposed permutation operators. 
A permutational representation is a map $V:S(n)\rightarrow \Hom(\mathcal{(\mathbb{C}}^{d})^{\otimes n})$ of the symmetric group $S(n)$, where $n=N+1$, in the space $\mathcal{H\equiv (\mathbb{C}}^{d})^{\otimes n}$ defined as
	\be
	\label{repV}
	\forall \pi \in S(n)\qquad V(\pi ).|e_{i_{1}}\>\otimes |e_{i_{2}}\>\otimes
	\cdots \otimes |e_{i_{n}}\>:=|e_{i_{\pi ^{-1}(1)}}\>\otimes |e_{i_{\pi
			^{-1}(2)}}\>\otimes \cdots \otimes |e_{i_{\pi ^{-1}(n)}}\>,
	\ee
where the set $\{|e_{i}\>\}_{i=1}^{d}$ is an orthonormal basis of the space $\mathcal{\mathbb{C}}^{d}$, and $d$ stands for the dimension. We drop here the lower index in every $i$, since it labels only position of the basis in tensor product $(\mathbb{C}^{d})^{\otimes n}$.
%Since the representation $V$ is defined in a given basis of the space $\mathbb{C}^d$, it is a matrix representation, and 
In other words, the operators $V(\pi)$ just permute basis vectors according to the given permutation $\pi\in S(n)$. The representation $V$ extends to the
representation of the group algebra $\mathbb{C}[S(n)]:=\Span_{\mathbb{C}}\{V(\sigma ):\sigma \in S(n)\}\subset \Hom(\mathcal{(\mathbb{C}}^{d})^{\otimes n})$.
%\be
%\label{CSn}
%\mathcal{A}_{n}(d):= \Span_{\mathbb{C}}\{V(\sigma ):\sigma \in S(n)\}\subset %\Hom(\mathcal{(\mathbb{C}}^{d})^{\otimes n}).
%\ee
All irreducible representations (irreps) of the symmetric group $S(n)$ are labeled by so-called  \textit{partitions}. A partition $\alpha$ of a natural number $n$, which we denote as $\alpha \vdash n$, is a sequence of positive numbers $\alpha=(\alpha_1,\alpha_2,\ldots,\alpha_r)$, such that $\alpha_1\geq \alpha_2\geq \cdots \geq \alpha_r$ and $\sum_{i=1}^r\alpha_i=n.$
Every partition can be visualised as a \textit{Young frame} which is a collection of boxes arranged in left-justified rows. For illustration please see panel {\bf I} in Figure~\ref{YngBox}. It means that for every fixed number $n$, the number of Young frames determines the number of nonequivalent irreps of $S(n)$ in an abstract decomposition.
%A partition $\alpha$ of a natural number $n$, which we denote as $\alpha \vdash n$, is a sequence of positive numbers $\alpha=(\alpha_1,\alpha_2,\ldots,\alpha_r)$, such that
%\be
%\alpha_1\geq \alpha_2\geq \cdots \geq \alpha_r,\qquad \sum_{i=1}^r\alpha_i=n.
%\ee
%Every partition can be visualised as a \textit{Young frame} which a collection of boxes arranged in left-justified rows. For illustration please see panel A in Figure~\ref{YngBox}. 
\begin{figure}[h]
\centering
 \includegraphics[width=0.6\columnwidth,keepaspectratio,angle=0]{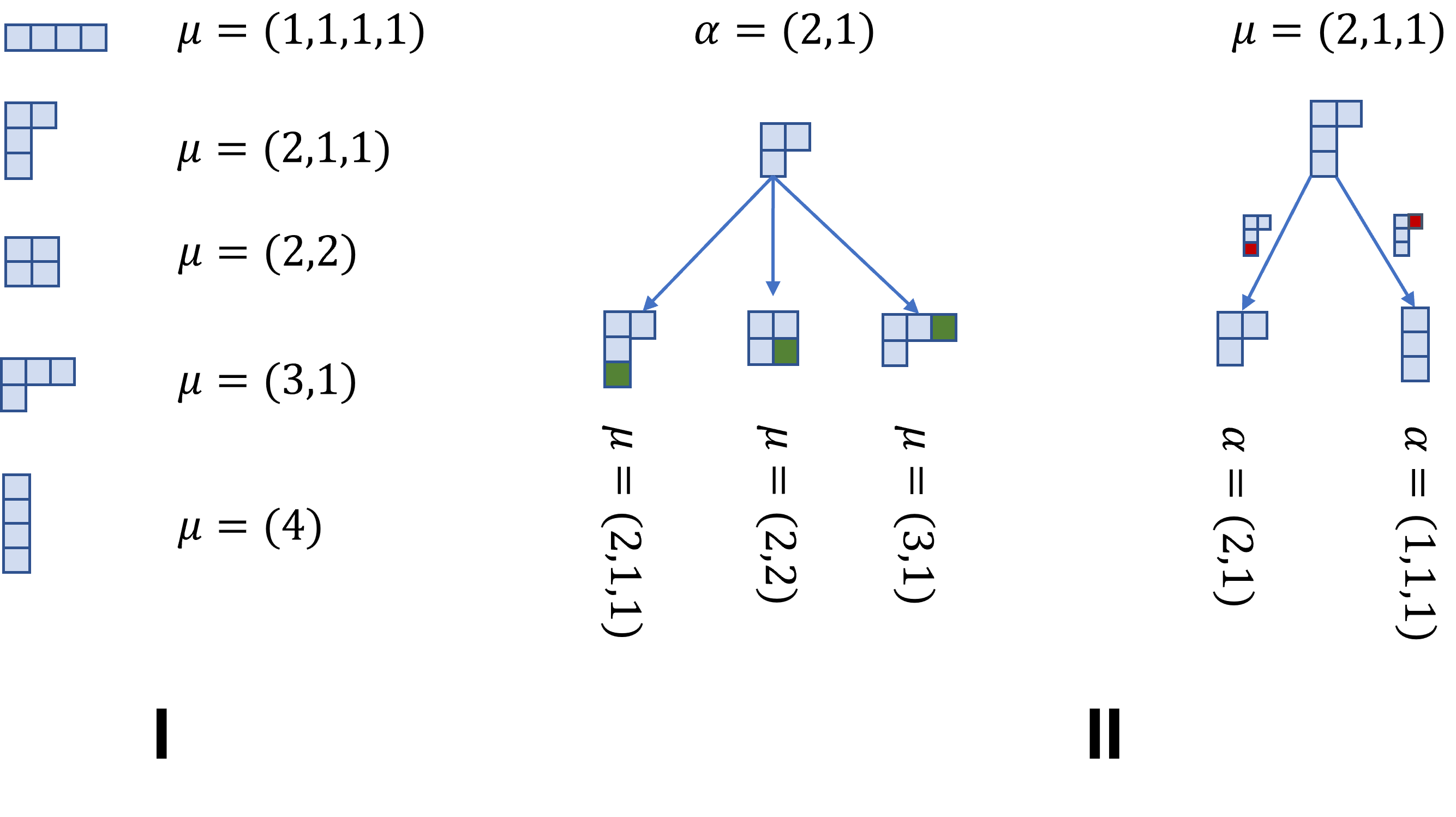}
\caption{Panel {\bf I} depicts five possible Young frames for $n=4$ corresponding  to all possible abstract irreducible representations of the group $S(4)$. Its representation space is $(\mathbb{C}^d)^{\otimes 4}$ and the only irreps that appear are those for which the height of corresponding Young frames is no larger than $d$. In particular, when one considers qubits ($d=2$) we have only three admissible frames: $(4),(3,1),(2,2)$. Panel {\bf II} presents all possible Young frames $\mu \vdash 4$ satisfying relation $\mu \in \alpha$ for $\alpha=(2,1)$. 
Green squares depict boxes that are added to the initial frame $\alpha$. On the right, we present all possible Young frames $\alpha \vdash 3$ satisfying relation $\alpha \in \mu$ for $\mu=(2,1,1)$. Boxes subtracted from the initial Young frame $\mu=(2,1,1)$ are shown in red.} 
\label{YngBox}
 \end{figure}
  If one fixes the representation space to $\mathcal{H\equiv (\mathbb{C}}^{d})^{\otimes n}$ then when we decompose $S(n)$ into irreps we take only Young frames $\alpha$ whose height $h(\alpha)$ is at most $d$.
Let us take now $\alpha \vdash (n-1)$ and $\mu \vdash n$. By writing $\mu \in \alpha$ we understand Young frames $\mu$ obtained from $\alpha$ by adding a single box. On the opposite, $\alpha \in \mu$ denotes Young frames $\alpha$ obtained from $\mu$ by removing a single box -- see panel {\bf II} in Figure~\ref{YngBox}.

Recall the celebrated Schur-Weyl duality~\cite{FultonSchur,Schur-Weyl}, which states that the diagonal action of the unitary group $\mathcal{U}(d)$ of invertible complex matrices and of the symmetric group $S(n)$ on $(\mathbb{C}^d)^{\otimes n}$ commute:
\be
\label{comm1}
[V(\sigma),U\otimes \cdots \otimes U]=0,
\ee
where $\sigma \in S(n)$ and $U\in \mathcal{U}(d)$.  In particular, it means that there exists a basis producing the decomposition
of $V(\pi)$ and $U^{\otimes n}$ into irreps simultaneously, and decomposition of the  tensor product space $(\mathbb{C}^d)^{\otimes n}$ as:
\be
\label{eq:SW}
(\mathbb{C}^d)^{\otimes n}=\bigoplus_{\substack{\alpha \vdash n \\ h(\alpha)\leq d}} \mathcal{U}_{\alpha}\otimes \mathcal{S}_{\alpha}.
\ee
In the above expression the symmetric group $S(n)$ acts non-trivially on the space $\mathcal{S}_{\alpha}$ and the unitary group $\mathcal{U}(d)$ acts non-trivially on the space $\mathcal{U}_{\alpha}$, labelled by the same partitions $\alpha$. 

For a given irrep $\alpha$ of $S(n)$, the space $\mathcal{U}_{\alpha}$ is a multiplicity space of dimension $m_{\alpha}$ (multiplicity of irrep $\alpha$), while the space $\mathcal{S}_{\alpha}$ is a representation space of dimension $d_{\alpha}$ (dimension of irrep $\alpha$).
With every subspace $\mathcal{U}_{\alpha}\otimes \mathcal{S}_{\alpha}$ we associate the \textit{Young projector}:
\be
\label{Yng_proj}
P_{\alpha}=\frac{d_{\alpha}}{n!}\sum_{\sigma \in S(n)}\chi^{\alpha}(\sigma^{-1})V(\sigma),\quad \text{with}\quad \tr P_{\alpha}=m_{\alpha}d_{\alpha},
\ee
where $\chi^{\alpha}(\sigma^{-1})$ is the irreducible character associated to the irrep indexed by $\alpha$.  To denote a matrix representation of an irrep of $\sigma \in S(n)$ indexed by a frame $\alpha$ we will write  $\varphi^{\alpha}(\sigma)$.

Equation~\eqref{eq:SW} gives a way to describe irreducible representations -- the minimal non-trivial blocks commuting with a diagonal action of the unitary group $\mathcal{U}(d)$, i.e. action of the form $U^{\otimes n}$, for some natural $n$, and $U\in \mathcal{U}(d)$. Such reduction, in addition to displaying the interior structure of the operators allows one to specialize the analysis on the whole space (which is typically very large and complex) to every small block separately thus significantly reducing the complexity of the problem (this is especially helpful when problems involve semidefinite programming). When working with the PBT we need to introduce different type of symmetries, however, still motivated by the Schur-Weyl duality discussed above. To describe them we first briefly discuss the properties of states and measurements used in all variants of PBT described in Subsection~\ref{prelim:A}.

Recall that a bipartite maximally entangled state is $U\otimes \overline{U}$ invariant~\cite{HorodeckiMPRFidelity}, where the bar denotes complex conjugation of an element $U$ of the unitary group $\mathcal{U}(d)$. It means that all signal states $\sigma_i^{A\B}$ from~\eqref{eq:signals} satisfy the following commutation rule:
\begin{equation}
\label{sym1}
\begin{split}
[U^{\otimes N}\otimes \overline{U},\sigma_i^{A\B}]=0,\quad \forall \ U\in \mathcal{U}(d),
\end{split}
\end{equation}
where $\overline{U}$ acts on $\B$, and $U^{\otimes N}$ acts on systems $A=A_1\cdots A_N$. Additionally, from the construction of the signal states $\sigma_i^{A\B}$  it follows that they are covariant with respect to the elements from the group $S(N)$, acting on first $N$ systems:
\begin{equation}
\begin{split}
V(\pi)\sigma_i^{A\B}V^{\dagger}(\pi)=\sigma_{\pi(i)}^{A\B},\quad \forall \ \pi\in S(N).
\end{split}
\end{equation}
In particular, choosing an arbitrary state from the set, say $\sigma_N^{A\B}$ all the others can be generated by acting on it with an element from the coset $S(N)/S(N-1)$,
whose elements in the representation $V$ are of the form $V[(i,N-1)]$, for $i=1,\ldots,N-1$. The same kind of covariance with respect to $S(N)$ and $U^{\otimes N}\otimes \overline{U}$ also holds for all measurements $\{\Pi_i^{AC}\}_{i=1}^N$ in PBT.
Finally, since the operator $\rho$ from~\eqref{eq:measurements} is a sum over all possible signal states it also exhibits symmetries described in~\eqref{sym1} and in addition, it is invariant with respect to the action of elements from the group $S(N)$.

Satisfying the relation~\eqref{sym1} by operators describing all variants of PBT protocols means that they belong to the \textit{algebra of partially transposed permutation operators}, where the partial transposition is taken with respect to the last $n-$th system~\cite{Moz1,MozJPA,Stu1}:
\be
	\A:= \Span_{\mathbb{C}}\{V'(\sigma ):\sigma \in S(n)\},
	\ee
 where for simplicity $'$ denotes the partial transposition under consideration.
The algebra $\A$ is no longer a group algebra since for example, one has $V'V'=dV'$, while in $\mathbb{C}[S(n)]$ we have $VV=\mathbf{1}$.
Elements $V'(\sigma)$ commute with skew-diagonal action of $U^{\otimes (N-1)}\otimes \overline{U}$, where $U\in \mathcal{U}(d)$. This is an analog of the relation from~\eqref{comm1} between the permutation group $S(n)$ and unitary group $\mathcal{U}(d)$. This means we should expect analogous decomposition of the space $(\mathbb{C}^d)^{\otimes n}$ to \eqref{eq:SW}, when studying objects from $\A$. 
 Due to the results contained in~\cite{Studzinski2017} we know that a similar  decomposition exists, with the corresponding version of Young projectors from~\eqref{Yng_proj} on irreducible spaces. We denote these projectors as $F_{\mu}(\alpha)$, and they are labelled by two types of Young diagrams $\alpha \vdash (n-2)$ and $\mu \vdash (n-1)$, such that $\mu\in \alpha$. 
In particular, it was shown in~\cite{Studzinski2017} that the operator $\rho$ from~\eqref{eq:measurements} decomposes in terms of irreducible projectors $F_{\mu }(\alpha )$ as
\begin{equation}
\label{rho_spectral}
\rho=\sum_{\alpha \vdash N-1}\sum_{\mu \in \alpha}\lambda_{\mu}(\alpha)F_{\mu}(\alpha)
\end{equation}
with non-zero eigenvalues $\lambda_{\mu}(\alpha)$ of the form
	\begin{equation}
\label{llambda}
\lambda_{\mu}(\alpha)=\frac{1}{d^N}\gamma_{\mu}(\alpha)=\frac{N}{d^{N}}\frac{m_{\mu}d_{\alpha}}{m_{\alpha}d_{\mu}}.
\end{equation}

The above decomposition allows for easy calculation of the inversion $\rho^{-1}$ necessary for having explicit form of the measurements from~\eqref{eq:measurements}.
We also introduce notation $\gamma_{\mu^*}(\alpha)$, meaning that for a given diagram $\alpha \vdash (n-2)$ we choose such $\mu\in\alpha$ for which $\gamma_{\mu}(\alpha)$ from~\eqref{llambda} is maximal, i.e.
\begin{equation}
\label{lambdamax}
\gamma_{\mu^*}(\alpha):=\max_{\mu\in\alpha}\gamma_{\mu}(\alpha).
\end{equation}

\section{Port-based teleportation with minimal requirements}\label{minimalrecs}

To introduce the minimal set of requirements that defines a PBT protocol, consider the following sequence of steps outlined in the box below.
\begin{tcolorbox}
[ title={PBT protocol ${\cal P}$}]
{\bf INPUT} 

 $n\ge 0$, a shared $2n-qubit$ state $\rho_{AB}^{(n)}$ between sender $A$ and receiver $B$; a state $\psi_C$ to be teleported; a set measurements $\{{\cal M}_{A,i}\}_{i=1}^k$, where without loss of generality we have $k=n+1$. The instantiated protocol is denoted as ${\cal P}_n(\{{\cal M}_{A,i}\}_{i=1}^k,\rho_{AB}^{(n)},\psi_C)$.

{\bf ALGORITHM}

\begin{itemize}
\item Alice performs a measurement ${\cal M}_{A,i}$ on $\rho_{AB}^{(n)}\otimes \psi_C$, obtaining outcome $i\in [1,\ldots, k]$.
\item Alice sends the index $i$ to Bob by classical channel; 
\end{itemize}

{\bf OUTPUT}

If $i\in [1,\ldots,k]$, then the teleported state is $\rho_{B_{i}}^{(n)}$. Otherwise, return "FAIL".

\end{tcolorbox}\label{PBTalgorithm}
With each ${\cal P}_n(\{{\cal M}_{A,i}\}_{i=1}^k,\rho_{AB}^{(n)},\psi_C)$ we associate a two-parameter estimate $Q({\cal P}_n(\{{\cal M}_{A,i}\}_{i=1}^k,\rho_{AB}^{(n)})) =(F(\sigma_A,\rho^{(n)}_{B_{i}}), p)$ that describes the performance of a teleported state. The first parameter characterises the quality of teleported state and the second --  the success of the teleportation. Please notice that the parameter $Q$ does not depend on the input state $\psi_C$. In the general situation, one could take $k\geq n+1$ measurements, but since only first $n$ of them give contribution to $F$ and $p$ the rest can be treated as a one big measurement labelled by index $i=n+1$.

\begin{definition}
A PBT protocol $\cal P$ from the above algorithm from the box is {\it convergent} if $Q({\cal P})\to (1,1)$ as $n\to\infty$. 
\end{definition}

\subsection{Minimal PBT protocol (mPBT)}
\label{intrpolated_prot}

In what follows, we introduce the convergent PBT protocol that while structurally similar to both pPBT and dPBT. However, while sharing many of the common building blocks with the d(p)PBT this `minimal' PBT protocol cannot be derived from either of them. The performance of our protocol interpolates between figures of merit -- the fidelity and probability of success -- between deterministic and probabilistic scheme respectively. Our protocol cannot be improved to get perfect fidelity or the probability of success with finite amount of resources (quantified as a number number of shared ports $N$). However, this protocol maintains exponentially faster convergence to unity, compared to the optimal pPBT. To construct the measurement operators for mPBT, we exploit the SRM that are used in dPBT, but crucially we omit an extra error term $(1/N)\Delta$  defined below expression~\eqref{eq:measurements}:
\begin{equation}
\label{srm}
\forall 1\leq i\leq N \quad \Pi_i^{AC}=\rho^{-1/2}\sigma_i^{AC}\rho^{-1/2}\quad \text{where}\quad \rho=\sum_{i=1}^N\sigma_i^{AC}.
\end{equation}

It is clear that the POVM elements do not sum up to identity on the whole space $(\mathbb{C}^d)^{\otimes n}$, where $n=N+1$. To fix that and thus recover their proper probabilistic interpretation we must add an additional operator denoted here $M_0^{AC}$. This additional POVM element, in this variant, equals exactly to $\Delta$. While operationally, it corresponds the failure  of the teleportation process, we show in Appendix~\ref{summary} that it bears no resemblance to any known $M_0$ (failure) operators which were used for probabilistic teleportation schemes. Note that we drop the constraint from equation~\eqref{ort}, so the state is not teleported faithfully in this probabilistic scheme (see later discussion in this section).

The probability of success and fidelity of our protocol are given below:
\begin{theorem}
\label{pSRM}
The probability of success $p_{succ}$ in mPBT, with the non-optimised resource state, when one uses square-root measurements $\{\Pi_i^{AC}\}_{i=1}^N$ from~\eqref{srm} has the form:
\begin{equation}
\label{pqudits}
p_{succ}=\frac{1}{d^{N+1}}\sum_{\alpha \vdash N-1}\sum_{\mu \in \alpha}d_{\mu}m_{\alpha},
\end{equation}

where $m_{\alpha}, d_{\mu}$ denote multiplicities and dimension of irreps of $S(N-1)$ and $S(N)$ respectively in the Schur-Weyl duality. The success probability for qubits $p_{succ}$ has the form:
\begin{equation}
\label{pqubits}
p_{succ}=1-\frac{N+2}{2^{N+1}}.
\end{equation}

\end{theorem}

\begin{proof}
To prove expression for the probability of success $p_{succ}$ we use equation~\eqref{psucca} with $\widetilde{O}_A=\mathbf{1}_A$ and explicit form of measurements from~\eqref{srm}:
\begin{equation}
\label{psa}
p_{succ}=\frac{1}{d^{N+1}}\sum_{i=1}^N\tr\left[\Pi_i^{AC}\right]=\frac{1}{d^{N+1}}\sum_{i=1}^N\tr\left[\rho^{-1/2}\sigma_i^{AC}\rho^{-1/2}\right]=\frac{1}{d^{N+1}}\tr(\rho^{-1/2}\rho\rho^{-1/2})=\frac{1}{d^{N+1}}\tr(\mathbf{1}_{\operatorname{supp}(\rho)}),
\end{equation}
where $\mathbf{1}_{\operatorname{supp}(\rho)}$ is the identity operator on the support of $\rho$ from~\eqref{srm}.
From paper~\cite{Studzinski2017} we know that $\mathbf{1}_{\operatorname{supp}(\rho)}$ has decomposition into sum of projectors $F_{\alpha}(\mu)$ on irreducible spaces of the algebra $\mathcal{A}'_n(d)$:
\begin{equation}
\label{34}
\mathbf{1}_{\operatorname{supp}(\rho)}=\sum_{\alpha \vdash N-1}\sum_{\mu \in \alpha}F_{\alpha}(\mu), \quad \text{so}\quad \tr(\mathbf{1}_{\operatorname{supp}(\rho)})=\sum_{\alpha \vdash N-1}\sum_{\mu \in \alpha}d_{\mu}m_{\alpha},
\end{equation}
since $\tr F_{\alpha}(\mu) =d_{\mu}m_{\alpha}$.
This leads us to the first statement from the theorem. To prove the second part note that in the qubit case all Young diagrams are up to two rows, and they are always of the form $\mu=(N-l,l)$, where $1\leq l\leq \left \lceil \frac{N}{2}-1\right\rceil$. It is possible to derive closed-form expressions for dimensions and multiplicities of the corresponding irreps. From~\cite{fulton_harris}, we know that for any partition of the shape $\mu=(N-l,l)$ the following expression for the corresponding of irrep dimension and multiplicity:
\be
	d_\mu =\binom{N}{l}-\binom{N}{l-1}=\frac{(N-2l+1)}{N+1}\binom{N+1}{l},\quad m_\mu=N-2l+1.
	\label{eq:m_d_alpha}
\ee
However, in expression~\eqref{34} we have to consider two types of irreps, irreps $\alpha \vdash N-1$ and irreps $\mu \vdash N$, for which we have the relation $\mu\in \alpha$. The trace in~\eqref{34} can be written as
\begin{equation}
\tr(\mathbf{1}_{\operatorname{supp}(\rho)})=\sum_{\mu\vdash N}\sum_{\alpha \in \mu}d_{\mu}m_{\alpha},
\end{equation}
where $\alpha \in\mu$ denotes partitions $\alpha \vdash N-1$ obtained from partition $\mu \vdash N$ by removing a single box. We have two types of partitions $\alpha$. The first type is when $N-1=2k$, for $\left \lceil \frac{N}{2}-1\right\rceil$, then the corresponding Young tableaux has two equal rows. In this case, we can add a box only to the first row. In the second type of partitions, when $N-1=2k+1$. In this case, we can add a single box to the first or second row. Using this observation and expressions~\eqref{eq:m_d_alpha},  we can write terms $d_{\mu}m_{\alpha}$ for Young tableaux satisfying relation $\alpha \in\mu$:
\begin{enumerate}[I)]
\item if $N-1=2k$, then 
\begin{equation}
\tr(\mathbf{1}_{\operatorname{supp}(\rho)})=N^{2}+\sum_{1\leq l<k}(N+1)\binom{%
N}{l}\frac{(N-2l)^{2}}{(N+1-l)(l+1)}+\frac{2}{N+1}\binom{N+1}{k},
\end{equation}

\item if $N-1=2k+1,$ then 
\begin{equation}
\tr(\mathbf{1}_{\operatorname{supp}(\rho)})=N^{2}+\sum_{1\leq l<k}(N+1)\binom{%
N}{l}\frac{(N-2l)^{2}}{(N+1-l)(l+1)}.
\end{equation}
\end{enumerate}
These two expressions can be combined to get one closed expression of the form
\begin{equation}
\tr(\mathbf{1}_{\operatorname{supp}(\rho)})=\sum_{l=0}^{\left \lceil \frac{N}{2}-1\right\rceil} (N+1){N \choose l} \frac{(N-2l)^2}{(N+1-l)(l+1)}.
\end{equation}
Finally, to show expression~\eqref{pqubits} we use the following chain of equalities:
\begin{equation}
\begin{split}
p_{succ}&=\frac{1}{2^{N+1}}\sum_{l=0}^{\left \lceil \frac{N}{2}-1\right\rceil} (N+1){N \choose l} \frac{(N-2l)^2}{(N+1-l)(l+1)}=\frac{N+1}{2^{N+2}}\sum_{l=0}^{N}\frac{(N-2l)^2}{(l+1)(N-l+1)}\\
&=\frac{1}{2^{N+2}}\frac{1}{N+2}\sum_{l=0}^N (N-2l)^2{N+2\choose l+1}.
\end{split}
\end{equation}
Now, changing range of the sum and subtracting proper terms, we arrive to
\begin{equation}\label{explicitformula}
\begin{split}
p_{succ}&=\frac{1}{2^{N+2}}\frac{1}{N+2}\left[\sum_{l=0}^{N+2}\left({N+2\choose l}(N-2(l-1))^2\right)-2(N+2)^2\right]\\
&=\frac{1}{2^{N+2}}\frac{1}{N+2}\sum_{l=0}^{N+2}\left[\left({N+2\choose l}((N+2)^2-4(N+2)l+4l^2)\right)-2(N+2)^2\right]\\
&=\frac{1}{2^{N+2}}\frac{1}{N+2}\left[\sum_{l=0}^{N+2}\left({N+2\choose l}(N+2)^2-4(N+2)l{N+2\choose l}+4l^2{N+2\choose l}\right)-2(N+2)^2\right]\\
&=\frac{1}{2^{N+2}}\frac{1}{N+2}\left[2^{N+2}(N+2)^2-4(N+2)^22^{N+1}+4(N+2+(N+2)^2)2^N-2(N+2)^2\right]\\
&=1-\frac{N+2}{2^{N+1}},
\end{split}
\end{equation}
where in the third line we use identities $\sum_{l=0}^N{N\choose l}=2^N$, $\sum_{l=0}^{N}l{N\choose l}=N2^{N-1}$, and $\sum_{l=0}^Nl^2{N\choose l}=N(N+1)2^{N-2}$. This finishes the proof.
\end{proof}
In the right panel of figure~\ref{mPBT vs pPBT} we present actual values for $p_{succ}$ obtained from Theorem~\ref{pSRM} for qubits and we compare it with optimal values for the pPBT introduced in~\cite{ishizaka_quantum_2009}. Additionally, in Figure~\ref{mPBTd vs pPBTd} we plot efficiency characteristic for higher dimensions ($d=4$) by exploiting group theoretical expression  from~\eqref{pqudits} and compare it with optimal values for the pPBT for $d\geq 2$ derived in~\cite{Studzinski2017}.
\begin{figure}[!h]
   \centering
        \includegraphics[scale=0.65]{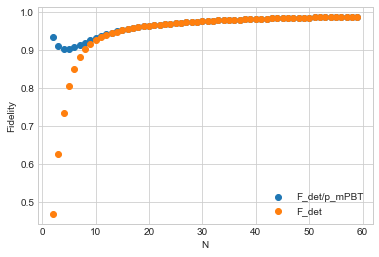}
          \includegraphics[scale=0.65]{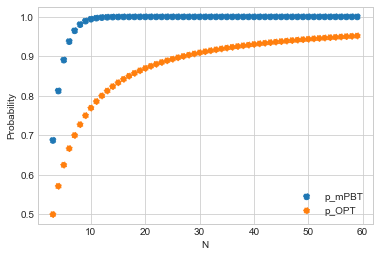}
\caption{Qubit case. The left panel: The orange dots show entanglement fidelity of the teleported state when Alice uses SRM measurements in the non-optimal dPBT. The blue dots represent entanglement fidelity obtained by mPBT protocol from expression~\eqref{fint2}. In the regime of $N>8$ both protocols give similar fidelities, while for $N\leq 8$ the mPBT protocol outperforms dPBT significantly. In the case of mPBT fidelity behaves non-monotonically, i.e. it decreases to the value 0.902 for $N=4$ and then it starts growing, achieving asymptotically unit value. The right panel: The  blue dots depict values of $p_{succ}$ obtained in mPBT scheme performed by SRM measurements from Theorem~\ref{pSRM}. It visibly outperforms the optimal pPBT, which is $p_{succ}=1-\frac{3}{N+3}$ - the orange dots. }
  \label{mPBT vs pPBT}
\end{figure}

\begin{figure}[!h]
 \centering
        \includegraphics[scale=0.65]{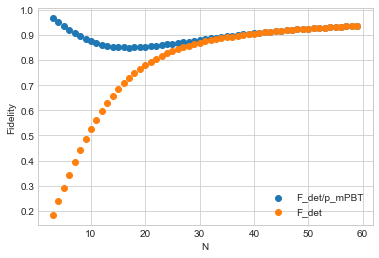}
          \includegraphics[scale=0.65]{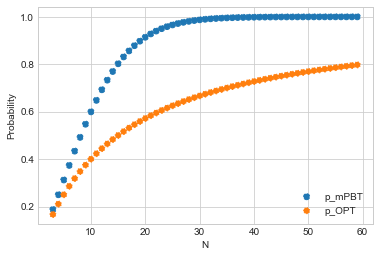}
          \caption{Behaviour of the fidelities (left panel) and probability of success (right panel) for $d=4$. Description of the colors is the same as it is in Figure~\ref{mPBT vs pPBT}. In this case, the mPBT scheme performed by SRM measurements from Theorem~\ref{pSRM} we outperform the optimal qudit pPBT, for which the probability of success is $p_{succ}=1-\frac{d^2-1}{N+d^2-1}$~\cite{Studzinski2017}. When one considers resulting fidelities, we see that the mPBT scheme is more efficient up to $\sim30$ ports when one compares it with the non-optimal dPBT with SRM measurements.} 
          \label{mPBTd vs pPBTd}
\end{figure}

The next theorem shows the best achievable performance in this setting:  Alice optimises both measurements and the resource state using an operation $\widetilde{O}_A$ is given as: 
\begin{equation}
    \label{expOA1}
    \widetilde{O}_A=\sqrt{d^N}\sum_{\mu \vdash N}\frac{v_{\mu}}{\sqrt{d_{\mu}m_{\mu}}}P_{\mu},
    \end{equation}
    where $v_{\mu}\geq 0$ are entries of an eignevector corresponding to a maximal eigenvalue of the teleportation matrix $M_F$ used for computation of entanglement fidelity in optimal PBT~\cite{StuNJP}. 
\begin{theorem}
The probability of success in mPBT, with the optimised resource state is given by
\begin{equation}
p_{succ}=\frac{1}{d}\sum_{\alpha \vdash N-1}\sum_{\mu \in \alpha}\frac{m_{\alpha}}{m_{\mu}}v^2_{\mu},
\end{equation}
where $m_{\alpha}, m_{\mu}$ denote multiplicities of irreps of $S(N-1)$ and $S(N)$ respectively in the Schur-Weyl duality.
\end{theorem}

\begin{proof}
To prove expression for the probability of success $p_{succ}$ we use equation~\eqref{psucca} with $\widetilde{O}_A$ from~\eqref{expOA1} and explicit form of measurements presented in~\eqref{srm}:
\begin{equation}
\begin{split}
p_{succ}=\frac{1}{d^{N+1}}\sum_{i=1}^N\tr(\widetilde{O}_A\Pi_i^{AC}\widetilde{O}_A^{\dagger}).
\end{split}
\end{equation}
Since we know that $\sum_i \Pi_i^{AC}=\mathbf{1}_{\operatorname{supp}(\rho)}=\sum_{\alpha}\sum_{\mu\in\alpha}F_{\mu}(\alpha)$ we write:
\begin{equation}
p_{succ}=\frac{1}{d^{N+1}}\sum_{\alpha \vdash N-1}\sum_{\mu\in\alpha}\tr (\widetilde{O}_A\widetilde{O}_A^{\dagger}F_{\mu}(\alpha))=\frac{1}{d^{N+1}}\sum_{\alpha \vdash N-1}\sum_{\mu\in\alpha}\frac{m_{\alpha}}{m_{\mu}}\tr (\widetilde{O}_A\widetilde{O}_A^{\dagger}P_{\mu}).
\end{equation}
In the second equality we use Lemma 10 from~\cite{Studzinski2017} stating that $\tr_n F_{\mu}(\alpha)=\frac{m_{\alpha}}{m_{\mu}}P_{\mu}$, where $P_{\mu}$ is a Young projector on irrep labelled by $\mu \vdash N$. 
Plugging the explicit form of $\widetilde{O}_A$ given in~\eqref{expOA1} into the above, using orthogonality property for Young projectors $P_{\nu}P_{\mu}=\delta_{\nu\mu}P_{\mu}$, and taking into account Observation~\ref{obs} from Appendix~\ref{summary} we find:
\begin{equation}
p_{succ}=\frac{1}{d}\sum_{\alpha \vdash N-1}\sum_{\mu\in\alpha}\frac{m_{\alpha}v_{\mu}^2}{m_{\mu}^2d_{\mu}}\tr(P_{\mu})=\frac{1}{d}\sum_{\alpha \vdash N-1}\sum_{\mu \in \alpha}\frac{m_{\alpha}}{m_{\mu}}v^2_{\mu},
\end{equation}
since $\tr(P_{\mu})=d_{\mu}m_{\mu}$. This finishes the proof.
\end{proof}
We thus evaluated the probability of success $p_{succ}$ by exploiting two types of measurements used by Alice for qudits. Unlike the qubit case of Theorem~\ref{pSRM}, it does not lend itself to a nice closed form.

We now turn to entanglement fidelity of the above protocols. The entanglement fidelity, where one teleports a half of the maximally entangled state, is given by
\begin{equation}
\label{Fvsp}
F(\mathcal{N}_{C\rightarrow \B})\equiv F=\frac{1}{p_{succ}}\tr\left[P^+_{BD}\left(\mathcal{N}_C\otimes \mathbf{1}_D\right)P^+_{CD}\right]=\frac{1}{d^2p_{succ}}\sum_{i=1}^N\tr(\Pi_i^{AC}\sigma_i^{AC}),
\end{equation}
where $\sigma_i^{AC}=(1/d^{N-1})(\mathbf{1}_{\overline{A}_i}\otimes P^+_{A_iC})$, and $\overline{A}_i$ denotes all systems but $i-$th.
We expect to teleport the state faithfully, so we must have $F(\mathcal{N}_{C\rightarrow \B})=1$. Note that the term $\frac{1}{d^2}\sum_{i=1}^N\tr(\Pi_i^{AC}\sigma_i^{AC})$ is the entanglement fidelity $F_{det}$ of the respective protocol in deterministic scheme, since we have that $\tr(M_0^{AC}\sigma_i^{AC})=0$. This observation allows us to express entanglement fidelity in the minimal PBT scheme (see equation~\eqref{F_det}) as:
\begin{equation}
\label{fint}
F=\frac{F_{det}}{p_{succ}}.
\end{equation}
Expressions for $F_{det}$ in (non-)optimal PBT are known and presented respectively in Theorem 12 in~\cite{Studzinski2017} and Proposition 32 in~\cite{StuNJP}. For the self-consistence we provide these expressions in Appendix~\ref{summary}, please see expressions~\eqref{F_non} and~\eqref{Fopt} respectively, with corresponding qubit forms in~\eqref{F_nonqubit} and~\eqref{Foptqubit}. In the case of qubits the probability of success in mPBT is given through expression~\eqref{pqubits} from Theorem~\ref{pSRM} and fidelity $F_{det}$ is given through~\eqref{F_nonqubit}, so~\eqref{fint} reads as:
\begin{equation}
\label{fint2}
F=\frac{\frac{1}{2^{N+3}}\sum_{k=0}^N\left(\frac{N-2k-1}{\sqrt{k+1}}+\frac{N-2k+1}{\sqrt{N-k+1}}\right)^2\binom{N}{k}}{1-\frac{N+2}{2^{N+1}}}.
\end{equation}
In the left panel of figure~\ref{mPBT vs pPBT} we compare the actual fidelity of mPBT from~\eqref{fint}  with the fidelity of the teleported state in the non-optimal PBT scheme~\eqref{fint2}, performed by SRM measurements.

Figure~\ref{PBT_landscape} describes the known landscape of admissible protocols and their interrelation when parties exploit maximally entangled resource state. A similar plot can be make for the optimised port-based teleportation schemes, the only thing which would change is the scaling in the number of ports $N$, but the general shape stays the same.
\begin{figure}[h]
\centering
\includegraphics[scale=1.41]{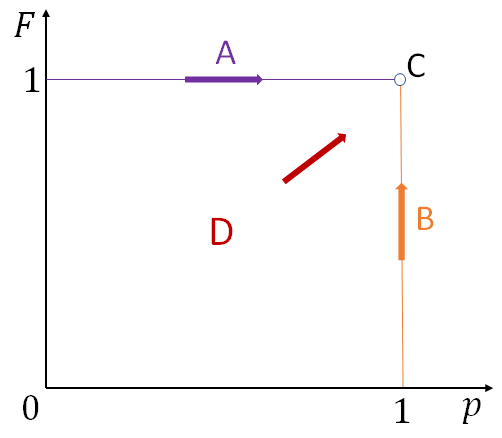}
\caption{Four regions for port-based teleportation protocols; $N$ is the number of ports. Each point defined by a pair of coordinates that corresponds to a protocol with probability of success and fidelity $(F, p)$. {\bf A} (purple): the family of pPBT protocols with $(1,1-O(1/\sqrt{N}))$; {\bf B} (orange):the family of dPBT protocols with $(1-O(1/N),1)$; {\bf C} (black): impossibility region for port-based teleportation protocols with finite resources~\cite{ishizaka2008asymptotic}; {\bf D} (red): new port-based teleportation protocol with minimal assumptions and scaling $(\left[1-O(1/N)\right]/p_{succ}, p_{succ})$, where $p_{succ}=1-(N+2)/2^{N+1}$. The arrows indicates asymptotic convergence of considered protocols.}
\label{PBT_landscape}
\end{figure}

\section{Converting PBT protocols}\label{converting}
Different types of port-based teleportation require individual mathematical analysis and in this section we discuss the possibility of converting probabilistic version of PBT into the deterministic one. First, we prove a general theorem giving an explicit relation between entanglement fidelity and probability of success in such conversion. Next, we illustrate the theorem by presenting expressions for the entanglement fidelity of teleported state in the case of non-optimal pPBT and its optimal version, where explicit form of Alice's measurements is known. We also argue that the reverse conversion, i.e. converting dPBT to probabilistic one is possible only under certain constraints.

\begin{theorem}
\label{transforming}
Every pPBT scheme, with $N$ ports each of dimension $d$, and the probability of success $p_{succ}$, can be turned into deterministic
with explicit entanglement fidelity of the form:
\begin{equation}
\label{eq:transforming}
F=p_{succ}+\frac{1}{d^2N}\tr\left(M_0^{AC}\rho\right),
\end{equation}
where $M_0^{AC}$ is measurements corresponding to the failure of probabilistic scheme, operator $\rho$ is defined through expression~\eqref{srm}.   
\end{theorem}

\begin{proof}
Every pPBT scheme requires a set of $N+1$ measurements $\{M_0^{AC},M_1^{AC},\ldots,M_N^{AC}\}$, where the effect $M_0^{AC}$ corresponds to the failure of the teleportation process. Additionally, to get $F=1$ in probabilistic scheme, we require that all the measurements for $1\leq i \leq N$ satisfy the relation~\eqref{ort}. To design corresponding deterministic scheme, where teleported state appears on one through $N$ ports on the Bob's side we perform the mapping:
\begin{equation}
\label{POVM_mapping}
\Pi_1^{AC} =M_1^{AC}+\frac{1}{N}M_0^{AC},\quad \Pi_1^{AC} =M_2^{AC}+\frac{1}{N}M_0^{AC}, \quad \ldots \quad ,\Pi_N^{AC} = M_N^{AC}+\frac{1}{N}M_0^{AC}.
\end{equation}
Then the teleportation channel $\mathcal{N}(\psi_C)$ is of the form
\begin{equation}
\mathcal{N}(\psi_C)=\sum_{i=1}^N\tr_{A\bar{B}_iC}\left[ \sqrt{\Pi_i^{AC}}\left(\Psi_{AB}^+\ot \psi_{C} \right)\sqrt{\Pi_{i}^{AC}}^{\dagger}\right]_{B_i\rightarrow \B}.
\end{equation}
Now, if $\psi_{C}$ is a half of maximally entangled state $P^+_{CD}$, the entanglement fidelity $F(\mathcal{N})$ of the channel $\mathcal{N}$ is
\begin{equation}
\label{Fp}
\begin{split}
F(\mathcal{N})&=\tr\left[P^+_{BD}(\mathcal{N}_{C}\otimes \mathbf{1}_B)(P^+_{CD})\right]=\frac{1}{d^2}\sum_{i=1}^N\tr\left(\Pi_i^{AC}\sigma_i^{AC}\right)=\frac{1}{d^2}\sum_{i=1}^N\tr\left[\left(M_i^{AC}+\frac{1}{N}M_0^{AC}\right)\sigma_i^{AC}\right]\\
&=\frac{1}{d^2}\sum_{i=1}^N\tr\left(M_i^{AC}\sigma_i^{AC}\right)+\frac{1}{Nd^2}\tr\left(M_0^{AC}\rho\right),
\end{split}
\end{equation}
because $\rho=\sum_i \sigma_i^{AC}$. Since we start from POVMs for probabilistic scheme, they must satisfy relation~\eqref{ort}, which implies:
\begin{equation}
\tr\left(M_i^{AC}\sigma_i^{AC}\right)=\frac{1}{d^{N-1}}\tr\left(M_i^{AC}\right).
\end{equation}
Plugging the above into~\eqref{Fp} and using~\eqref{psucca} we get the result. 
\end{proof}

The conversion of a mPBT to dPBT protocol can be also viewed differently. Namely, every probabilistic PBT protocol leads to a deterministic one with the fidelity $F=F_{succ}p_{succ}+F_{fail}(1-p_{succ})$, where $F_{succ}=1$, and $F_{fail}=1/d^2$ is entanglement fidelity when one fails with the transmission process. Similar approach of turning every pPBT into dPBT shortly discussed in the paper~\cite{majenz2}. In~\cite{majenz2} authors consider a type of conversion when whenever Alice fails in probabilistic scheme she sends to Bob a random port index. However, our approach is a little bit different from the one presented above, since in Theorem~\ref{transforming} we do not have anymore possibility of the failure - see construction of the corresponding POVMs in eq.~\eqref{POVM_mapping}. Our construction divides POVM $M_0^{AC}$ into $N$ parts and add them to every POVM $M_i^{AC}$ getting pure deterministic scheme.
Summarising, Theorem~\ref{transforming} except the direct formula for the entanglement fidelity of such protocol gives also an algorithm for its construction by the explanation how to construct respective measurements.

From expression~\eqref{eq:transforming} we see that the resulting fidelity depends on two factors -- probability of success in probabilistic scheme and overlap between POVM $M_0^{AC}$ corresponding to failure and the state $\rho$ from~\eqref{eq:measurements}. 
In the general case it is impossible to say much about the overlap, since the operator $M_0^{AC}$ depends heavily on an architecture of a given probabilistic scheme. However, when we stick with a particular probabilistic scheme we can explicitly evaluate the entanglement fidelity $F$ in Theorem~\ref{transforming}. In what follows we illustrate how Theorem~\ref{transforming} works in practice for known pPBT schemes. We shall consider both, non-optimal and optimal pPBT for an arbitrary number of ports and their arbitrary dimension, which have been analysed in~\cite{Studzinski2017}. Our calculations are based on the latter paper. A  short summary of the existing results on pPBT together with some group theoretic methods developed in papers~\cite{Studzinski2017,StuNJP} is provided in Section~\ref{summary}.

\begin{lemma}
\label{npPBT}
The entanglement fidelity $F$ given in Theorem~\ref{transforming} for the non-optimised resource state reads:
\begin{equation}
\label{eq:npPBT}
F=\frac{1}{d^N}\sum_{\alpha \vdash N-1}\frac{m_{\alpha}d_{\alpha}}{\gamma_{\mu^*}(\alpha)}+\frac{1}{d^2}\left(1-\frac{1}{d^N}\sum_{\alpha \vdash N-1}\sum_{\nu \in \alpha}m_{\nu}d_{\alpha}\frac{\gamma_{\nu}(\alpha)}{\gamma_{\mu^*}(\alpha)}\right),
\end{equation}
where the numbers $\gamma_{\nu}(\alpha),\gamma_{\mu^*}(\alpha)$ are given by~\eqref{llambda} and~\eqref{lambdamax} respectively. The symbols $\mu,\nu$ denote  Young diagrams with $N$ boxes obtained from a Young diagram $\alpha$ of $N-1$ boxes by adding a single box. The quantities $m_{\alpha},m_{\nu}$ are multiplicities of irreps of permutation groups $S(N-1), S(N)$ respectively in the Schur-Weyl duality, and $d_{\alpha}$ is dimension of an irrep of the permutation group $S(N-1)$ in the Schur-Weyl duality.
\end{lemma}

\begin{proof}
The first term $p_{succ}$ in expression~\eqref{eq:transforming} is known, see Theorem 3 in~\cite{Studzinski2017}, and equals to
\begin{equation}
p_{succ}=\frac{1}{d^N}\sum_{\alpha \vdash N-1}\frac{m_{\alpha}d_{\alpha}}{\gamma_{\mu^*}(\alpha)}. 
\end{equation}
It remains to evaluate the term $\tr(M_0^{AC}\rho)$. Using the summation rule for POVMs
\begin{equation}
\label{povm_sum}
\sum_{i=0}^N M_i^{AC}=\mathbf{1}_{AC},
\end{equation}
where $\mathbf{1}_{AC}$ is an identity operator on $(\mathbb{C}^d)^{\otimes (N+1)}$, 
 $M_0^{AC}=\mathbf{1}_{AC}-\sum_{i=1}^NM_i^{AC}$. This gives us
\begin{equation}
\begin{split}
\tr(M_0^{AC}\rho)=\tr(\rho)-\sum_{i=1}^N\tr(M_i^{AC}\rho)=N-N\tr(M_N^{AC}\rho).
\end{split}
\end{equation}
The last equality follows from the fact that every operator $\sigma_i^{AC}$ in the sum $\rho=\sum_i\sigma_i^{AC}$ is normalised, and fact that the trace and $\rho$ are invariant under action of elements from the coset $S(N)/S(N-1)$.
Using explicit form of POVM $M_N^{AC}$ given in~\eqref{nonoptmeasurements} in the Appendix~\ref{summary}, and decomposition of the state $\rho$ into irreducible projectors of the algebra $\A$ given in~\eqref{rho_spectral} we get:
\begin{equation}
\begin{split}
\tr(M_N\rho)&=\sum_{\alpha \vdash N-1}\sum_{\beta \vdash N-1}\sum_{\nu \in \beta}\frac{\lambda_{\nu}(\beta)}{\gamma_{\mu^*}(\alpha)}\tr(P_{\alpha}V'F_{\nu}(\beta))=\sum_{\alpha \vdash N-1}\sum_{\beta \vdash N-1}\sum_{\nu \in \beta}\frac{\lambda_{\nu}(\beta)}{\gamma_{\mu^*}(\alpha)}\tr(P_{\alpha}V'P_{\nu}P_{\beta})\\
&=\sum_{\alpha \vdash N-1}\sum_{\nu \in \alpha}\frac{\lambda_{\nu}(\alpha)}{\gamma_{\mu^*}(\alpha)}\tr(P_{\alpha}P_{\nu}V')=\sum_{\alpha \vdash N-1}\sum_{\nu \in \alpha}\frac{\lambda_{\nu}(\alpha)}{\gamma_{\mu^*}(\alpha)}\tr(P_{\alpha}P_{\nu})=\sum_{\alpha \vdash N-1}\sum_{\nu \in \alpha}m_{\nu}d_{\alpha}\frac{\lambda_{\nu}(\alpha)}{\gamma_{\mu^*}(\alpha)}\\
&=\frac{1}{d^N}\sum_{\alpha \vdash N-1}\sum_{\nu \in \alpha}m_{\nu}d_{\alpha}\frac{\gamma_{\nu}(\alpha)}{\gamma_{\mu^*}(\alpha)}.
\end{split}
\end{equation}
The second equality follows from Theorem 1 and Fact 13 in the paper~\cite{Studzinski2017}. The third equality follows from the orthogonality property for Young projectors, saying that $P_{\beta}P_{\alpha}=\delta_{\beta\alpha}P_{\alpha}$. 
The fourth equality follows from the observation that only operator $V'$ acts non-trivially on last $n-$th system, so it can be traced out with respect to it, giving us the identity operator acting on first $N$ systems. 
The fifth equality follows from the fact that only Young projector $P_{\nu}$ acts non-trivially on $N-$th system, so we can apply Corollary 10 from~\cite{stud2020A}, and compute the partial trace. Using orthogonality property for Young projectors an fact that $\tr P_{\alpha}=d_{\alpha}m_{\alpha}$ the result follows.
\end{proof}

It is apparent from Lemma the final result is hard to  analyse analytically. However, every quantity from~\eqref{eq:npPBT} can be computed numerically.  
Next, we prove Lemma~\ref{npPBT} for the optimal pPBT, i.e. when Alice optimises simultaneously over measurements and resource state, and show that the final expression for the entanglement fidelity is in very elegant and compact form.
\begin{lemma}
\label{pPBT}
The entanglement fidelity $F$ given through Theorem~\ref{transforming} for the optimised resource state reads:
\begin{equation}
F=p_{succ}=1-\frac{d^2-1}{N+d^2-1},
\end{equation}
where $p_{succ}$ denotes the probability of success in the optimal pPBT. 
\end{lemma}

\begin{proof}
The main idea of the proof follows the proof of Lemma~\ref{npPBT}. However, here in equation~\eqref{povm_sum} we have to consider POVMs from~\eqref{measurement_opt} with the operation $O_A$ given in the same expression. All the POVMs must satisfy the relation $\sum_{i=0}^NM_i^{AC}=\mathbf{1}_{AC}$,
where $\mathbf{1}_{AC}$ is the identity operator acting on $(\mathbb{C}^d)^{\otimes (N+1)}$. Thus we write $\tr(M_0^{AC}\rho)$ as in the proof of Lemma~\ref{npPBT}:
\begin{equation}
\label{cos}
\tr(M_0^{AC}\rho)=\tr(\rho)-N\tr(M_N^{AC}\rho).
\end{equation}
The first term in the above expression is $\tr(\rho)=N$.

We now compute the second term in~\eqref{cos} using explicit form of POVM $M_N^{AC}$ from~\eqref{measurement_opt}:
\begin{equation}
\label{p2}
\begin{split}
\tr(M_N^{AC}\rho)&=\frac{1}{d}\sum_{\alpha \vdash N-1}u_{\alpha}\sum_{\beta \vdash N-1}\sum_{\mu \in \alpha}\lambda_{\mu}(\beta)\tr(P_{\alpha}V'F_{\mu}(\beta))=\frac{1}{d}\sum_{\alpha \vdash N-1}u_{\alpha}\sum_{\beta \vdash N-1}\sum_{\mu \in \alpha}\lambda_{\mu}(\beta)\tr(V'P_{\alpha}P_{\beta}P_{\mu})\\
&=\frac{1}{d}\sum_{\alpha \vdash N-1}u_{\alpha}\sum_{\mu \in \alpha}\lambda_{\mu}(\alpha)\tr(P_{\alpha}P_{\mu})=\frac{1}{d}\sum_{\alpha \vdash N-1}u_{\alpha}\sum_{\mu \in \alpha}\lambda_{\mu}(\alpha)m_{\mu}d_{\alpha}=\sum_{\alpha \vdash N-1}\sum_{\mu\in \alpha}\frac{d_{\alpha}}{d_{\mu}}\frac{m_{\mu}^2}{\sum_{\nu}m_{\nu}^2}\\
&=\sum_{\mu \vdash N}\left(\sum_{\alpha\in\mu}d_{\alpha}\right)\frac{1}{d_{\mu}}\frac{m_{\mu}^2}{\sum_{\nu}m_{\nu}^2}=1.
\end{split}
\end{equation}
All the steps except the last equality in~\eqref{p2} follow the proof of Lemma~\ref{npPBT}. First, we use the observation that $\sum_{\alpha\in\mu}d_{\alpha}=d_{\mu}$. Then we apply Proposition 25 from~\cite{Studzinski2017}.
Finally, combining expressions $\tr(\rho)=N$ and~\eqref{p2} we conclude that $\tr(M_0^{AC}\rho)=0$. Next, using expression for probability of success $p_{succ}$ in this variant of pPBT, which is $p_{succ}=1-\frac{d^2-1}{N+d^2-1}$ due to Theorem 4 in~\cite{Studzinski2017} we get the result.
\end{proof}

The above result implies that we cannot derive a deterministic scheme from probabilistic variant with entanglement fidelity scaling better than $1-\mathcal{O}(1/N)$ in the number of ports $N$. It follows from the fact that the probability of success in optimal pPBT scales as $1-\mathcal{O}(1/N)$, see~\cite{Studzinski2017}. This is due to the fact that in every probabilistic scheme we have to add an additional constraint: to ensure that the teleported state is transmitted faithfully, we have to ensure that all measurements corresponding to the probability of success satisfy~\eqref{ort}. It means that designing, say, the optimal probabilistic scheme one does not optimise over all possible space of POVMs, but over their proper subset. This restriction is one of the reasons responsible for different scaling in (non-)optimal probabilistic and deterministic protocols respectively. This implies that one cannot turn an arbitrary deterministic protocol into probabilistic one, with better scaling than $1-\mathcal{O}(1/N)$, since measurements of such protocol do not satisfy requirement~\eqref{ort}. In particular, we cannot turn optimal dPBT discussed in~\cite{ishizaka_quantum_2009,StuNJP} into probabilistic scheme with $F=1$ for finite $N$.
The only way for such conversion to be feasible is when one defaults to `interpolated' protocols, described in Section~\ref{intrpolated_prot}, where neither entanglement fidelity $F$ nor probability of success $p_{succ}$ equal to one for the finite number of ports $N$.
 %This differs from our scheme, since our goal is to design truly deterministic scheme, where Alice does not have to take into account any possibility of failure. This is because the POVM operator $M_0^{AC}$ corresponding to the procedure failure is split into $N$ equal parts and added to the rest of POVMs eliminating the possibility of failure altogether.

The conversion between PBT protocols is illustrated below on Figure~\ref{sec:fids}.

\begin{figure}[h]
\includegraphics[scale=0.51]{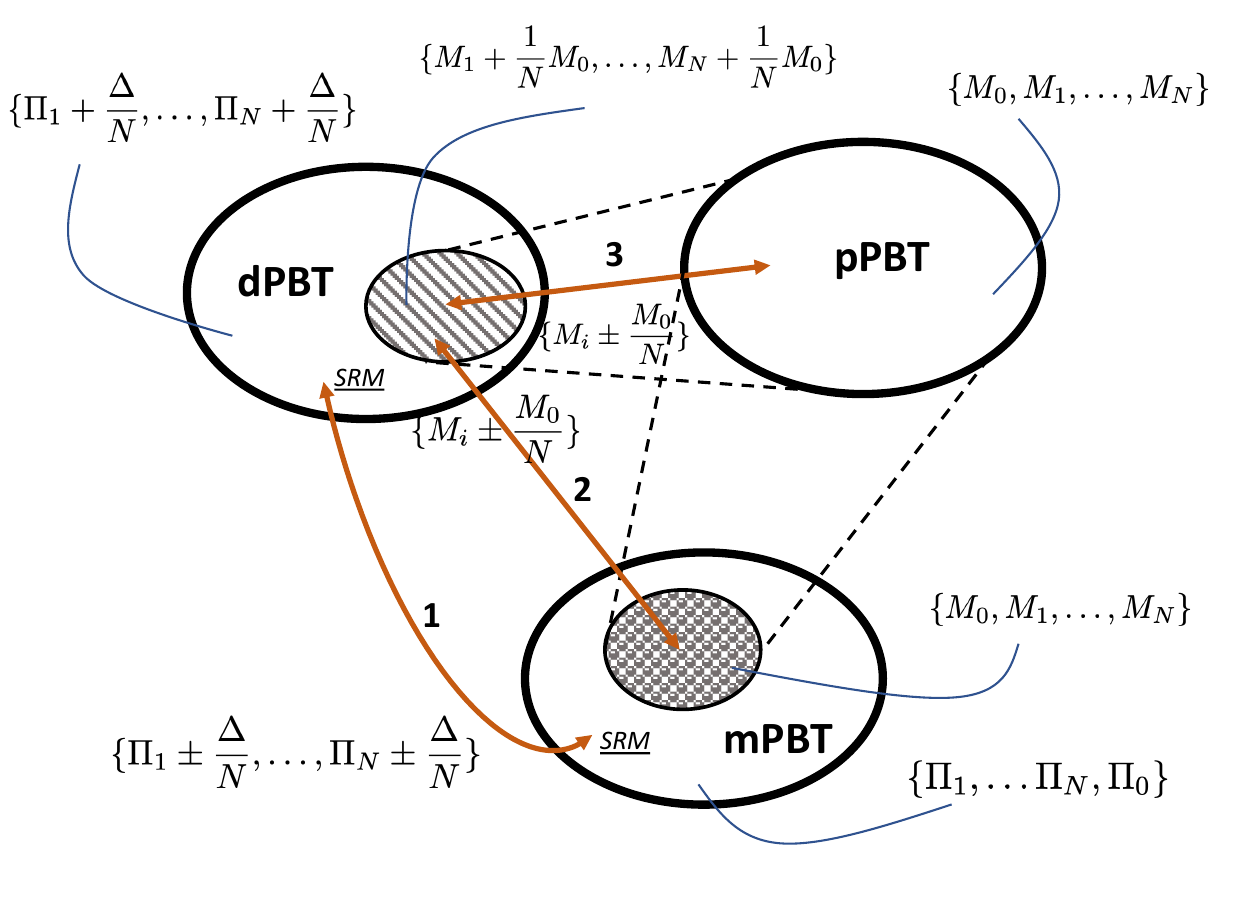}
\caption{
Some of the measurement operators for mPBT are fixed to be $ 1\leq i\leq N \quad \Pi_i^{AC}=\rho^{-1/2}\sigma_i^{AC}\rho^{-1/2}\quad \text{where}\quad \rho=\sum_{i=1}^N\sigma_i^{AC}$, with $\sigma_i^{AC}$ defined in Eq.~\eqref{eq:signals}. In general, there may exist other measurements that yield PBT protocols that satisfy requirements of 
Section~\ref{minimalrecs}.
Every mPBT protocol from the non-shaded region can be converted into the protocol in the non-shaded dPBT by adding $\Delta/N$ to each of the POVM elements, where $\Delta=\mathbf{1}-\sum_{i=1}^N\Pi_i$  (arrow 1). Any known pPBT protocol can be converted to a (shaded) subset of known dPBT protocols by taking all the POVM elements $M_1,\ldots M_N$ that correspond to successful teleportation (omitting $M_0$) and transform each of them as follows: $M_i = M_i + M_0/N$. This operation is invertible (arrow 3). Lastly, any mPBT protocol from the shaded area can be converted to shaded subset of dPBT protocols.
}
\end{figure}
\section{PBT resource state comparison}
\label{sec:fids}
%\subsection{Fidelity between optimal states in probabilistic and deterministic PBT.}

Despite the fact that all PBT protocols share operational similarities, their underlying resource states are rather different. In this section we discuss properties of the resource states in all PBT protocols by evaluating their closeness. 
We show that by starting from maximally entangled states, the the optimisation operator $O_A$ applied by Alice has a large effect on their mutual distances. To quantify the distance, we will use the square-root fidelity $\sqrt{F}$~\cite{Nielsen-Chuang}. See Appendix~\ref{SRfid} for details.

Let us denote the resource states in optimal probabilistic and optimal deterministic scheme by $|\Psi\>_{AB},|\widetilde{\Psi}\>_{AB}$ respectively. We distinguish optimization operators on Alice's side from Eq.~\eqref{resource} by writing $O_A,\widetilde{O}_A$ respectively. 
\begin{lemma}
\label{Fbetween}
The square fidelity $\sqrt{F}$ between the resource states in optimal pPBT $|\Psi\>_{AB}$ and optimal dPBT $|\widetilde{\Psi}\>_{AB}$, each with $N$ ports of dimension $d$ equals to:
\begin{equation}
\label{eq:Fbetween}
\sqrt{F}\left(|\Psi\>_{AB},|\widetilde{\Psi}\>_{AB}\right)=\sum_{\mu \vdash N}\frac{v_{\mu}m_{\mu}}{\sqrt{\sum_{\nu \vdash N}m_{\nu}^2}}, 
\end{equation}
where $m_{\mu},m_{\nu}$ are multiplicities of irreps of $S(N)$ in the Schur-Weyl duality.  The numbers $v_{\mu}\geq 0$ are entries of an eignevector corresponding to a maximal eigenvalue of the teleportation matrix $M_F$.
\end{lemma}
\begin{proof}
The general form of the optimal resource state in pPBT reads:
\be
\label{resource}
|\Psi\>_{AB}=(O_A \otimes \mathbf{1}_B)|\Psi^+\>_{AB}=(O_A \otimes \mathbf{1}_B)|\psi^+\>_{A_1B_1}\otimes |\psi^+\>_{A_2B_2}\otimes \cdots \otimes |\psi^+\>_{A_NB_N},
\ee
where $O_A$ is taken from~\eqref{measurement_opt}, with normalisation constraint $\tr(O_A^{\dagger}O_A)=d^N$. The same holds for the optimal dPBT, but instead of $O_A$ we use $\widetilde{O}_A$ given in~\eqref{expOA}. The square-root fidelity $\sqrt{F}$ between the states $|\Psi\>_{AB},|\widetilde{\Psi}\>_{AB}$:
\begin{equation}
\label{eq:FF}
\begin{split}
\sqrt{F}\left(|\Psi\>_{AB},|\widetilde{\Psi}\>_{AB}\right)=\left|\<\Psi|\widetilde{\Psi}\>_{AB}\right|=\left|\tr\left(|\widetilde{\Psi}\>\<\Psi|_{AB}\right)\right|=\frac{1}{d^N}\tr\left(\widetilde{O}_AO_A\right).
\end{split}
\end{equation}
Plugging the explicit forms of the operators $\widetilde{O}_A,O_A$ from~\eqref{expOA} and \eqref{measurement_opt} respectively we arrive to
\begin{equation}
\sqrt{F}\left(|\Psi\>_{AB},|\widetilde{\Psi}\>_{AB}\right)=\sum_{\mu,\nu \vdash N}\sqrt{\frac{g(N)m_{\mu}}{d_{\mu}}}    \frac{v_{\nu}}{\sqrt{d_{\nu}m_{\nu}}} \tr\left(P_{\mu}P_{\nu}\right).
\end{equation}
Taking into account orthogonality relation for Young projectors $P_{\mu}P_{\nu}=P_{\mu}\delta_{\mu\nu}$ with trace property $\tr P_{\mu}=d_{\mu}m_{\mu}$ and plugging the explicit form of the function $g(N)$ from~\eqref{measurement_opt_coeff} completes the proof.
\end{proof}

The numbers $v_{\mu}\geq 0$ in expression~\eqref{eq:Fbetween} are not known to admit closed form, since one does not know analytical expressions for eigenvectors of the teleportation matrix for $d>2$ and $d<N$ (see~\cite{StuNJP}). However, when $d=2$ the matrix $M_F$ is tri-diagonal and we know analytical expressions for its eigenvalues and eigenvectors due to~\cite{Losonczi1992}. In what follows, we will invoke results from~\cite{ishizaka_quantum_2009}, where the form of the operators $\widetilde{O}_A, O_A$ was evaluated in terms of the spin angular momentum of the $N-$spin system. In this representation the operators are proportional to identity on subspaces with fixed quantum number $j$, which runs from $j_{\min}=0(1/2)$ when $N$ is even (odd):
\begin{equation}
\label{eq:qubitOA}
\widetilde{O}_A=\sum_{j=j_{\min}}^{N/2}\sqrt{\gamma(j)}\mathds{1}(j),\qquad O_A=\sum_{j=j_{\min}}^{N/2}\sqrt{\zeta(j)}\mathds{1}(j),
\end{equation}
for known positive numbers $\gamma(j),\zeta(j)$ which we describe later. The above form of the optimal Alice's operations allows us to formulate the qubit version of Lemma~\ref{Fbetween}:
\begin{lemma}
\label{Fbetween2}
The square fidelity $\sqrt{F}$ between the resource states $|\Psi\>_{AB}$ and $|\widetilde{\Psi}\>_{AB}$ in optimal qubit pPBT and optimal dPBT respectively, each with $N$ ports equals to:
\begin{equation}
\label{eq:Fbetween}
\sqrt{F}\left(|\Psi\>_{AB},|\widetilde{\Psi}\>_{AB}\right)=\frac{2}{N+2}\sqrt{\frac{6}{(N+1)(N+3)}}\sum_{j=j_{\min}}^{N/2}(2j+1)\operatorname{sin}\left(\frac{\pi (2j+1)}{N+2}\right),
\end{equation}
where $j$ runs from $j_{\min}=0(1/2)$ when $N$ is even (odd). The fidelity $\sqrt{F}\left(|\Psi\>_{AB},|\widetilde{\Psi}\>_{AB}\right)$ converges to $\frac{\sqrt{6}}{\pi}\approx 0.778$ when $N\rightarrow \infty$ and takes maximal value $\sqrt{F}\left(|\Psi\>_{AB},|\widetilde{\Psi}\>_{AB}\right)=0.894$ which is attained for $N=2$.
\end{lemma}

\begin{proof}
The proof is similar to that of Lemma~\ref{Fbetween} with the difference we now use qubit versions of the operators $O_A,\widetilde{O}_A$ in spin angular momentum representation given in~\eqref{eq:qubitOA}. The dimension and the multiplicity $d_{\mu}, m_{\mu}$ respectively depend on the quantum number $j$:
\begin{equation}
\label{djmj}
d_j=\frac{(2j+1)N!}{(N/2-j)!(N/2+j+1)!},\qquad m_j=2j+1.
\end{equation}
This together with the trace rule $\tr \mathds{1}(j)=m_jd_j$ allows us to deduce that the square-root fidelity $\sqrt{F}$ is of the form
\begin{equation}
\sqrt{F}\left(|\Psi\>_{AB},|\widetilde{\Psi}\>_{AB}\right)=\frac{1}{2^N}\sum_{j=j_{\min}}^{N/2}\sqrt{\gamma(j)\zeta(j)}d_jm_j.
\end{equation}
Plugging the explicit form of the coefficients $\gamma(j),\zeta(j)$ (expression (35) and (54) in~\cite{ishizaka_quantum_2009})
\begin{equation}
\label{eq:coeff}
\gamma(j)=\frac{2^{N+2}}{(N+2)m_jd_j}\operatorname{sin}^2\left(\frac{\pi m_j}{N+2}\right),\qquad \zeta(j)=\frac{2^Nh(N)m_j}{d_j},
\end{equation}
and observing that the value of the sinus is always positive in the allowed range of $j$, and with expression for $h(N)$ (equation above (54) in~\cite{ishizaka_quantum_2009}), which reads as
\begin{equation}
\label{h}
h(N)=\frac{6}{(N+1)(N+2)(N+3)},
\end{equation}
we obtain the first part of the statement.
%-->A simple estimate yields: 
% Using Mathematica software we can show that the series from the statement converges to
%\begin{equation}
%\sqrt{F}\left(|\Phi\>_{AB},|\widetilde{\Phi}\>_{AB}\right)= h() + O(1/N^3)%\frac{1}{N+2}\sqrt{\frac{6}{(N+1)(N+3)}}\operatorname{Cosec}\left(\frac{\pi}{N+2}\right)\left[N+3+\operatorname{Cosec}\left(\frac{\pi}{N+2}\right)\operatorname{sin}\left(\frac{\pi(N+3)}{N+2}\right)\right].
%\end{equation}
From the above expression one can deduce numerically the convergence of the fidelity to the value $0.778$ for $N\rightarrow \infty$, and maximal value of $F\left(|\Psi\>_{AB},|\widetilde{\Psi}\>_{AB}\right)=0.894$ attained for two ports.
\end{proof}

The situation is entirely different when one considers fidelity between the resource states in non-optimal and optimal PBT, taking probabilistic and deterministic scheme for comparison. In both non-optimal pPBT and dPBT the resource state before the optimization has the same form of $N$ copies of maximally entangled state. We start from the following lemma:
\begin{lemma}
\label{ovpPBT}
The square fidelity $\sqrt{F}$ between the resource states in non-optimal pPBT and its optimal version, each with $N$ ports of dimension $d$ equals to:
\begin{equation}
\label{71}
\sqrt{F}\left(|\Psi\>_{AB},|\Psi^+\>_{AB}\right)=\frac{1}{d^N}\sum_{\mu \vdash N}m_{\mu}\sqrt{g(N)d_{\mu}m_{\mu}},
\end{equation}
where $d_{\mu}, m_{\mu}$ are dimension and multiplicity of irreps of $S(N)$ in the Schur-Weyl duality, and $g(N)$ is given in ~\eqref{measurement_opt_coeff}. In the qubit case the expression~\eqref{71} has a form
\begin{equation}
\label{72}
\sqrt{F}\left(|\Psi\>_{AB},|\Psi^+\>_{AB}\right)=\sqrt{\frac{6}{2^N(N+1)(N+2)(N+3)}}\sum_{j=j_{\min}}^{N/2}\frac{(2j+1)^2}{\sqrt{\left(\frac{N}{2}-j\right)!\left(\frac{N}{2}+j+1\right)!}},
\end{equation}
where $j$ runs from $j_{\min}=0(1/2)$ when $N$ is even (odd). The  qubit fidelity from~\eqref{72} converges to 0 with $N\rightarrow \infty$.
\end{lemma}

\begin{proof}
We start from proving our statement for an arbitrary dimension $d$ of the port. Using explicit form of $O_A$ in optimal pPBT given in~\eqref{measurement_opt} w have:
\begin{equation}
\label{eq:FF}
\begin{split}
\sqrt{F}\left(|\Psi\>_{AB},|\Psi^+\>_{AB}\right)=\frac{1}{d^N}\tr\left(O_A\right)=\frac{1}{\sqrt{d^N}}\sum_{\mu \vdash N}\sqrt{\frac{g(N)m_{\mu}}{d_{\mu}}}\tr(P_{\mu})=\frac{1}{\sqrt{d^N}}\sum_{\mu \vdash N}m_{\mu}\sqrt{g(N)d_{\mu}m_{\mu}},
\end{split}
\end{equation}
since $\tr(P_{\mu})=d_{\mu}m_{\mu}$. To get expression for qubits we use the second expressions from equations~\eqref{eq:qubitOA} and~\eqref{eq:coeff}, and then the explicit form of the function $h(N)$ in~\eqref{h}, together with equations for $d_j,m_j$ in terms of quantum number $j$ from~\eqref{djmj}. To prove convergence in the qubit case, note that the factor in front of the sum over $j$ clearly goes to 0 with $N\rightarrow \infty$. The second factor can be bounded from the above as follows
\begin{equation}
\sum_{j=j_{\min}}^{N/2}\frac{(2j+1)^2}{\sqrt{\left(\frac{N}{2}-j\right)!\left(\frac{N}{2}+j+1\right)!}}\leq \frac{\frac{N}{2}(N+1)^2}{\left(\frac{N}{2}\right)!}\xrightarrow{N\rightarrow \infty}0.
\end{equation}
\end{proof}

For the completeness of our results, we include the corresponding lemma from~\cite{deg} and giving the value of fidelity between resource states in non- and optimal dPBT.
\begin{lemma}
\label{l:FPBT}
The square fidelity between the resource state in non-optimal and optimal dPBT with $N$ ports, each of dimension $d$ is given as:
\begin{equation}
\label{FPBT}
\sqrt{F}(|\Psi^+\>_{AB},|\Psi\>_{AB})=\frac{1}{\sqrt{d^N}}\sum_{\mu \vdash N} v_{\mu}\sqrt{d_{\mu}m_{\mu}},
\end{equation}
where $v_{\mu}$ are entries of an eigenvector corresponding to a maximal eigenvalue of the teleportation matrix $M_F$, $m_{\mu},d_{\mu}$ denote multiplicity and dimension of irreps of $S(N)$ in the Schur-Weyl duality, and $P_{\mu}$ is a respective Young projector. 
For qubits, the fidelity between the resources states is of the form:
\begin{equation}
\label{eq:l:FPBT}
\sqrt{F}(|\Psi^+\>_{AB},|\Psi\>_{AB})=\sqrt{\frac{N!}{2^{N-2}(N+2)}}\sum_{j=j_{\min}}^{N/2}\frac{(2j+1)\operatorname{sin}\frac{\pi(2j+1)}{N+2}}{\sqrt{(\frac{N}{2}-j)!(\frac{N}{2}+j+1)!}},
\end{equation}
where $j_{\min}=0 (1/2)$ when $N$ is even (odd). 
\end{lemma}
The results of the all lemmas from this section for $d=2$ and $d>2$ are presented in table~\ref{tab_res_overlaps} and table~\ref{tab_res_overlaps} respectively located in Appendix~\ref{appD}. In the same Appendix we include figure~\ref{fig:my_label} and figure~\ref{fig:my_label2} which illustrate statements of Lemma~\ref{Fbetween2}, Lemma~\ref{ovpPBT}, and Lemma~\ref{l:FPBT} for qubits.

\section{Sending more bits with fewer ebits: efficient Port-based superdense coding} \label{superdense}
We now turn to superdense coding protocols induced by PBT schemes. In the case of ordinary teleportation, the underlying channel is given by an identity channel, and sending a single qubit would result in 2 bits of classical information. 

A number of varying PBT protocols gives rise to an equal number of superdense coding schemes. 
%To formalize the description of these protocols, define $SD(\rho, {\cal M})$ to be a protocol where two parties share a quantum state $\rho$, the sender transmits a qudit and the receiver applies a measurement $\cal M$.

Suppose that Alice performs a measurement in dPBT, then the unnormalised post-measurement states $\widetilde{\chi}_{DB_i}$, for $1\leq i\leq N$ read:
\begin{equation}
\label{post_measurement}
\begin{split}
\widetilde{\chi}_{DB_i}&=\tr_{AC}\left[\Pi_i^{AC}\left(P^+_{CD}\ot P^+_{A_1B_1}\ot P^+_{A_2B_2}\ot \cdots \ot P^+_{A_NB_N}\right)\right]\\
&=\tr_{AC}\left[\Pi_i^{BD}\left(P^+_{CD}\ot P^+_{A_1B_1}\ot P^+_{A_2B_2}\ot \cdots \ot P^+_{A_NB_N}\right)\right]\\
&=\frac{1}{d^{N+1}}\Pi_i^{BD}.
\end{split}
\end{equation}
the second line is obtained by applying Lemma~\ref{techL1} from Appendix~\ref{appA} twice. We introduce normalised post-measurement state $\chi_{DB_i}$
\begin{equation}
\chi_{DB_i}:=\frac{\widetilde{\chi}_{DB_i}}{\tr \widetilde{\chi}_{DB_i}}=\frac{N}{d^{N+1}}\Pi_i^{BD},
\end{equation}
since $\tr \Pi_i^{BD}=\frac{d^{N+1}}{N}$ due to Theorem 5 in~\cite{deg}.

To estimate the performance of the superdense protocols, we introduce
\begin{equation}
\label{coefqik}
q_{i|k}:=d^2p_iF_{i|k},
\end{equation}
where terms $F_{i|k}$ represent fidelities between the post-measurement state $\chi_{DB_i}$ and maximally entangled state $P^+_{DB_k}$. These fidelities are of the form
\begin{equation}
F_{i|k}=\tr\left[\chi_{DB_i}P^+_{DB_k}\right]=\frac{N}{d^{N+1}}\tr\left[\Pi_i^{BD}P^+_{DB_k}\right]=\frac{N}{d^2}\tr\left[\Pi_i^{BD}\sigma_k^{BD}\right].
\end{equation}
In the last equality we used the definition of states $\sigma_k^{BD}=\frac{1}{d^{N-1}}\mathbf{1}_{\overline{B}_i}\ot P^+_{DB_k}$, where $\overline{B}_i$ denotes all systems $B$ but $i$. Consider two cases: $i=k$ and $i\neq k$. In the first case: 
\begin{equation}
F_{k|k}=\frac{N}{d^{2}}\tr\left[\Pi_k^{BD}\sigma_k^{BD}\right]=F.
\end{equation}
The above follows from the fact that $\tr\left[\Pi_k^{BD}\sigma_k^{BD}\right]$ does not depend on index $1\leq k\leq N$ and $F=\frac{1}{d^2}\sum_{i=1}^N\tr\left[\Pi_k^{BD}\sigma_k^{BD}\right]$, which is entanglement fidelity in dPBT. We see that $F_{k|k}\rightarrow 1$ with $N\rightarrow \infty$, since the same holds for entanglement fidelity $F$ in dPBT. 
We turn to computing $F_{i|k}$ for $i\neq k$. Form the relation $\sum_i q_{i|k}=1$, where the coefficients are defined through~\eqref{coefqik}, we write
\begin{equation}
\begin{split}
&q_{k|k}+\sum_{i\neq k}q_{i|k}=1,\\
&d^2p_kF_{k|k}+\sum_{i\neq k}d^2p_iF_{i|k}=1.
\end{split}
\end{equation}
In PBT schemes $p_i=1/N$, since all the ports are equally probable, and $F_{i|k}=\widetilde{F}$ for all $i\neq k$ due to the covariance property, so
\begin{align}
    &F+(N-1)\widetilde{F}=\frac{N}{d^2},\\
    &\widetilde{F}=\frac{N}{d^2(N-1)}-\frac{F}{N-1}=\frac{1}{d^2\left(1-\frac{1}{N}\right)}-\frac{F}{N-1}.
\end{align}
We see that $\widetilde{F} \rightarrow \frac{1}{d^2}$ for $N\rightarrow \infty$. An explicit expression for the entanglement fidelity $F$ in PBT can be turned into an explicit expression for $\widetilde{F}$.
We now present a protocol that beats the performance of the only known superdense protocol~\cite{ishizaka_remarks_2015} derived from the non-optimized dPBT protocol with fidelity~\cite{beigi_konig}:
\begin{equation}
\label{b1}
F\geq \frac{N}{N+d^2-1}.
\end{equation}
 The closed form for similar lower bounds is generally hard to compute, but we know that there exist PBT protocols with fidelity scaling as $1-O(1/N^2)$~\cite{ishizaka_quantum_2009,majenz2} (optimal qubit dPBT) or  non-optimal ones but better factor. For convenience, Table~\ref{table:entFPBT} lists the known expressions for fidelity derived in~\cite{majenz2}.
\begin{center}
\begin{table}[h!]
	\begin{tabular}{c|c}

Teleportation protocol & Entanglement fidelity $F$\\
\hline
	Non-optimised dPBT & $F=1-\frac{d^2-1}{4N}+O\left(N^{-3/2+\delta} \right) $ \\[0.1cm]

	Optimised dPBT & $F\geq 1-\frac{d^5+O\left(d^{9/2} \right) }{4\sqrt{2}N^2}+O\left(N^{-3} \right) $\\[0.1cm]

\hline

	\end{tabular}
\caption{Asymptotic behaviour of dPBT with arbitrary port dimension $d$ and port number $N$. All the results are taken from~\cite{majenz2}.}
	\label{table:entFPBT}
\end{table}
\end{center} 
We see that the bound~\eqref{b1} is weaker than those in Table~\ref{table:entFPBT}.
Using the expression for mutual information from~\cite{ishizaka_remarks_2015} and plugging in the value of fidelity from the non-optimized dPBT from table~\ref{table:entFPBT}, we get:
\begin{equation}
\label{IABour}
I(A:B)=\log_2\left(1-\frac{d^2-1}{4N}\right)+\frac{d^2}{N}\left(1-\frac{d^2-1}{4N}\right)\log_2(d^2).
\end{equation}
This function outperforms $I(A:B)$ from~\cite{ishizaka_remarks_2015} which exploits the bound from~\eqref{b1}.
Indeed, the function from~\eqref{IABour} achieves maximum for
\begin{equation}
N=\frac{3d^2(d^2-1)\log_e(d^2)+\sqrt{(d-d^3)^2(2d^2+d^2\log_e(d^2)-2)}}{2\left(1-d^2+4d^2\log_e(d^2)\right)}.
\end{equation}

Moreover, for any lower bound $F_*$ on fidelity in an arbitrary dPBT the maximum amount of information that the associated superdense coding protocol can transfer is given by:
\begin{equation}
    I(A:B)=\log_2(F_*)+\frac{d^2F_*}{N}\log_2(d^2).
\end{equation}

While we still used $F_*$ from the non-optimized dPBT above, we already get an improvement. We expect to have a dramatic improvement in the amount of communicated information and simultaneously the reduction of entanglement consumption when one uses bound $F_*$ from optimised dPBT protocols, for example by exploiting the second bound from Table~\ref{table:entFPBT}. 

\section{Discussion}
In this paper, we introduce a novel variant of the PBT protocol called the minimal PBT. This protocol meets the minimal set of requirements that define a feasible PBT scheme. We analyze its efficiency and show that it over-performs optimized pPBT even with the resource state in a form of $N$ pairs of maximally entangled states. In parallel, it offers the same efficiency as the pre-existing PBT schemes when one is interested in the entanglement fidelity of the transmission.
In the second part, we investigate the possibility of conversion between different types of PBT, namely, we focus on conversion between probabilistic schemes to deterministic ones, and vice versa. We present the general recipe for such conversion and we show how it applies to existing variants of pPBT and dPBT with their connection to mPBT. We also derive the efficiency of such converted schemes.
In the next part of the manuscript, we discuss the application of existing knowledge on the deterministic PBT to super-dense coding schemes, and we show the possibility of transmission of more classical bits with lower consumption of shared maximally entangled pairs (ports). Finally, we present a detailed analysis and comparison of the resource states in deterministic and probabilistic PBT by considering their mutual fidelities. We show that mutual fidelity between resource states decreases with the number of ports showing that the considered states become more distant in the trace norm.

We also leave two important open questions. The first one is to explore derived expressions for the efficiency of the mPBT protocol in the asymptotic limits when the number of ports $N$ tends to infinity, analogously as it was done in~\cite{christandl2021asymptotic}.  Applying similar reasoning we can rid off group theoretical parameters like dimensions and multiplicities of irreps under interest and investigate their asymptotic behavior in terms of local dimension and the number of ports. Another important topic is to investigate noise influence on the performance of all known variants of PBT protocols, including defined here the mPBT scheme. It is well known that the noise in the practical implementation of all quantum information protocols is unavoidable and only having a detailed analysis in a real-world scenario regime can tell us about the real potential of discussed in the literature schemes. This is a very important problem, especially in the context of the recent developments in the PBT area -- the first model of the PBT formalism in continuous variables~\cite{pereira2021characterising}. The study of the impact of noise is therefore most natural in this setting, especially from the point of view of possible implementations. 

\section*{Acknowledgements}
Authors thank Ng Chung Lok Andrew from Trinity College, Cambridge for finding the closed-form expression in Eq.~\eqref{explicitformula}. S.S. acknowledges support from the Royal Society University Research Fellowship. 
M.S. is supported through grant Sonatina 2, UMO-2018/28/C/ST2/00004 from the Polish National Science Centre. 

\appendix 
\section{Summary of known results for $d\geq 2$ for pPBT and dPBT} 
\label{summary}
We collect here  results concerning explicit form of measurements and Alice's optimising operations $O_A,\widetilde{O}_A$. See~\cite{Studzinski2017,StuNJP} for more detailed discussion. Due to the covariance property of the measuremets it suffices to present a single measurement, for example corresponding to $i=N$.
\begin{enumerate}
    \item \textit{Non-optimal pPBT.} In this case $O_A=\mathbf{1}_A$ with measurement $M_N^{AC}$ of the form (see Section 2.5.1 in~\cite{Studzinski2017}):
    \begin{equation}
    \label{nonoptmeasurements}
        M_N^{AC}=d\sum_{\alpha \vdash N-1}\frac{1}{\gamma_{\mu^*}(\alpha)}\left(P_{\alpha} \otimes P^+_{N,n}\right),
    \end{equation}
    where $P_{\alpha}$ is a Young projector introduced in~\eqref{Yng_proj}, $P^+_{N,n}$ denotes the maximally entangled state between respective systems, and 
    \begin{equation}
        \gamma_{\mu^*}(\alpha):=\min_{\mu\in\alpha}\frac{1}{\gamma_{\mu}(\alpha)}=\frac{1}{N}\min_{\mu\in\alpha}\frac{m_{\alpha}d_{\mu}}{m_{\mu}d_{\alpha}}.
    \end{equation}
    The minimisation is taken over all $\mu$ which can be obtained from given $\alpha$ by adding a single box, as it is described in Section~\ref{math_intro} and Figure~\ref{YngBox} within it. The quantity $\gamma_{\mu}(\alpha)$ can be easily connected with eigenvalues $\lambda_{\mu}(\alpha)$ from~\eqref{llambda} of the operator $\rho$ in~\eqref{rho_spectral}:
    \begin{equation}
        \lambda_{\mu}(\alpha)=\frac{1}{d^N}\gamma_{\mu}(\alpha).
    \end{equation}
    The probability of success $p_{succ}$ in this variant is given by the expression (see Theorem 3 in~\cite{Studzinski2017}):
    \begin{equation}
	\label{eq:p_succ}
	p_{succ}=\frac{1}{d^N}\sum_{\alpha \vdash N-k}m_{\alpha}^2\mathop{\operatorname{min}}\limits_{\mu\in\alpha}\frac{d_{\mu}}{m_{\mu}},
    \end{equation}
    where the minimum is taken over all Young frames $\mu$ which can be obtained from a given Young frame $\alpha \vdash N-1$ by adding a single box (see Figure~\ref{YngBox} for the details). 
    Quantities $m_{\alpha},m_{\mu}$, and $d_{\mu}$ denote multiplicities and dimension of irreducible representations of corresponding symmetric group in the Schur-Weyl duality.
    \item \textit{Optimal pPBT.} In this case Alice's optimising operation is non-trivial, however optimal measurements differ only by coefficients (see Section 2.5.2 in~\cite{Studzinski2017}):
    \begin{equation}
    \label{measurement_opt}
        O_A=\sqrt{d^N}\sum_{\mu \vdash N}\sqrt{\frac{g(N)m_{\mu}}{d_{\mu}}}P_{\mu},\qquad M_N^{AC}=\sum_{\alpha \vdash N-1}u_{\alpha}\left(P_{\alpha}\otimes P^+_{N,n}\right)
    \end{equation}
    with
    \begin{equation}
    \label{measurement_opt_coeff}
        g(N)=\frac{1}{\sum_{\nu\in N}m^2_{\nu}},\qquad u_{\alpha}=\frac{d^{N+1}}{N}\frac{g(N)m_{\alpha}}{d_{\alpha}}.
    \end{equation}
    The probability of success $p_{succ}$ in the optimal scheme is given by the a compact expression (see Theorem 4 in~\cite{Studzinski2017}):
    \begin{equation}
\label{pSuccd}
        p_{succ}=1-\frac{d^2-1}{N+d^2-1}.
    \end{equation}
    This expression for $d=2$ reduces to expression (56) derived in~\cite{ishizaka_quantum_2009}
    \begin{equation}
        p_{succ}=1-\frac{3}{N+3}.
    \end{equation}
    \item \textit{Optimal dPBT.} Alice's optimising operation in the optimal variant is of the form (see Proposition 32 in~\cite{StuNJP}):
    \begin{equation}
    \label{expOA}
    \widetilde{O}_A=\sqrt{d^N}\sum_{\mu \vdash N}\frac{v_{\mu}}{\sqrt{d_{\mu}m_{\mu}}}P_{\mu},
    \end{equation}
    where $v_{\mu}\geq 0$ are entries of an eignevector corresponding to the maximal eigenvalue of the teleportation matrix $M_F$ used for computation of entanglement fidelity in optimal PBT (see Section 4  in~\cite{StuNJP}). 
    
    In the case when parties exploit maximally entangled state as a resource and Alice applies SRM measurements to run the protocol, the corresponding entanglement fidelity equals to (see Theorem 12 in~\cite{Studzinski2017}):
    \begin{equation}
\label{F_non}
F_{det}=\frac{1}{d^{N+2}}\sum_{\alpha \vdash N-1}\left(\sum_{\mu \in \alpha}\sqrt{d_{\mu}m_{\mu}} \right)^2,
\end{equation}
where $d_{\mu},m_{\mu}$ denote dimension and multiplicity of irreps of $S(N)$ in the Schur-Weyl duality. In particular case of qubits we can use expression (29) from~\cite{ishizaka_quantum_2009}, which is of the form:
\begin{equation}
\label{F_nonqubit}
F_{det}=\frac{1}{2^{N+3}}\sum_{k=0}^N\left(\frac{N-2k-1}{\sqrt{k+1}}+\frac{N-2k+1}{\sqrt{N-k+1}}\right)^2\binom{N}{k}.
\end{equation}

When Alice optimises over a resource state and measurements the entanglement fidelity can be computed as:
\begin{equation}
\label{Fopt}
F_{det}=\frac{1}{d^2}||M_F^d||_{\infty},
\end{equation}
where $M_F^d$ is the principal minor of dimension $d$ of the teleportation matrix $M_F$ introduced in Section 4 of~\cite{StuNJP}. the symbol $||\cdot||_{\infty}$ denotes the infinity norm of a matrix. 
For qubits, the expression~\eqref{Fopt} reduces to (see expression (41) in~\cite{ishizaka_quantum_2009} and Section 5.3 in~\cite{StuNJP}):
\begin{equation}
\label{Foptqubit}
F_{det}=\operatorname{cos}^2\left(\frac{\pi}{N+2}\right).
\end{equation}
\end{enumerate}

\section{Several technical facts}
\label{appA}
Here we prove a statement used later in Section~\ref{sec:fids}.
\begin{observation}
\label{obs}
The operations $O_A, \widetilde{O}_A$ in optimal pPBT and dPBT respectively satisfy the chain of equalities
\begin{equation}
O_A=O_A^{\dagger}=\overline{O}_A=O_A^{T},\qquad \widetilde{O}_A=\widetilde{O}_A^{\dagger}=\overline{\widetilde{O}}_A=\widetilde{O}_A^T.
\end{equation}
\end{observation}
Indeed, the optimising operations are  given through linear combination with real coefficients of Young projectors $P_{\mu}$ defined in~\eqref{Yng_proj}, see~\cite{Studzinski2017,StuNJP}. 

\begin{lemma}\label{techL1}
Let $\{|i\>\}_{i=1}^d$ be a basis and $|\widetilde{\psi}_+\>=\sum_{i=1}^d|i\>\otimes|i\>$ be the unnormalised maximally entangled state, then for any operator $X$ we have
\begin{equation}
\left(\mathbf{1}\otimes X\right)|\widetilde{\psi}_+\>=\left(X^T\otimes \mathbf{1}\right)|\widetilde{\psi}_+\>,
\end{equation}
where $X^T$ denotes transposition of $X$ with respect to the basis $\{|i\>\}_{i=1}^d$.
\end{lemma}
This lemma can be proven by direct inspection.

\section{Square root fidelity and fidelity}
\label{SRfid}
To investigate the closeness of the resource states we use notion of the \textit{square root fidelity} $\sqrt{F}$~\cite{Nielsen-Chuang}, which for two arbitrary quantum states $\rho,\sigma$ is given as
\begin{equation}
\label{sqrtF}
\sqrt{F}(\rho,\sigma):=\tr\left(\sqrt{\sqrt{\rho}\sigma \sqrt{\rho}}\right).
\end{equation}
The connection of $\sqrt{F}$ with the Uhlmann's fidelity $F$~\cite{Uhlmann:1976,RJozsa} is given by taking a square of the expression~\eqref{sqrtF}:
\begin{equation}
\label{UhlmannF}
F(\rho,\sigma):=(\sqrt{F}(\rho,\sigma))^2=\left[\tr\left(\sqrt{\sqrt{\rho}\sigma \sqrt{\rho}}\right)\right]^2.
\end{equation}
When one quantum state is pure, i.e. $\rho_{\psi}:=|\psi\>\<\psi|$ and the second one $\sigma$ is mixed, then one has
\begin{equation}
F(\rho_{\psi},\sigma):=(\sqrt{F}(\rho_{\psi},\sigma))^2=|\<\psi|\sigma|\psi\>|^2.
\end{equation}
When two quantum states $\rho_{\psi},\sigma_{\phi}$ are pure, i.e. $\rho_{\psi}:=|\psi\>\<\psi|, \sigma_{\phi}:=|\phi\>\<\phi|$, which is exactly the case in this paper, the above expression reduces to
\begin{equation}
F(\rho_{\psi},\sigma_{\phi}):=(\sqrt{F}(\rho_{\psi},\sigma_{\phi}))^2=|\<\psi|\phi\>|^2.
\end{equation}
For two arbitrary density operators $\rho,\sigma$ we can define the trace distance $\delta(\rho,\sigma)$ as
\begin{equation}
    \delta(\rho,\sigma):=\frac{1}{2}\tr\left(|\rho-\sigma|\right),
\end{equation}
we can upper and lower bound it using the notion of the fidelity $F$ by Fusch-van de Graaf inequalities:
\begin{equation}
1-\sqrt{F(\rho,\sigma)}\leq \delta(\rho,\sigma) \leq \sqrt{1-F(\rho,\sigma)}.
\end{equation}
In the particular case, when both states are pure, then the above inequalities simplify and give:
\begin{equation}
\delta(\psi,\phi)=\sqrt{1-F(\psi,\phi)}.
\end{equation}

\section{PBT resource state comparison}
\label{appD}

\begin{figure}
    \centering
    \includegraphics[width=0.45\textwidth]{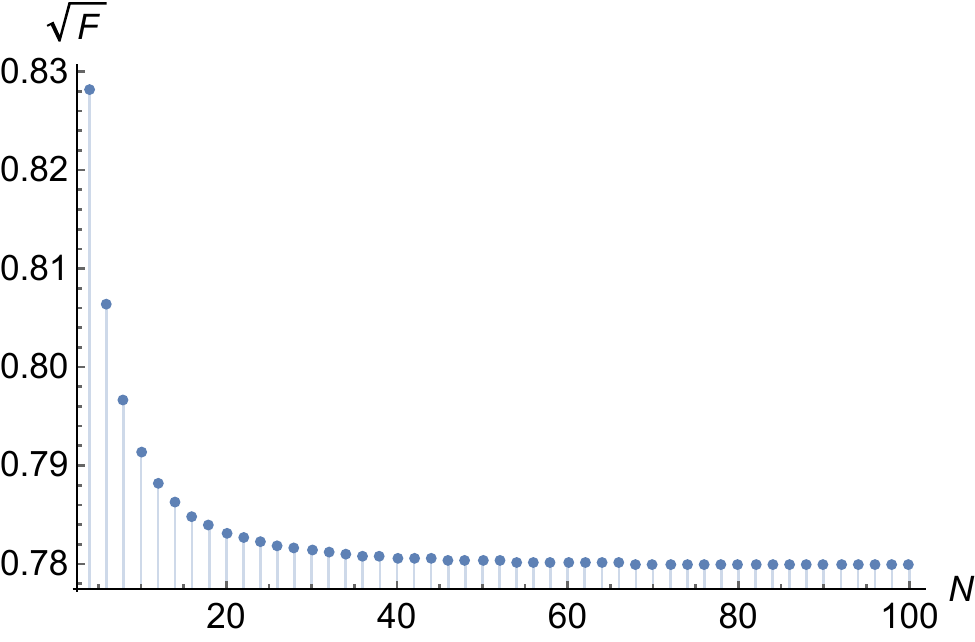}
    \includegraphics[width=0.45\textwidth]{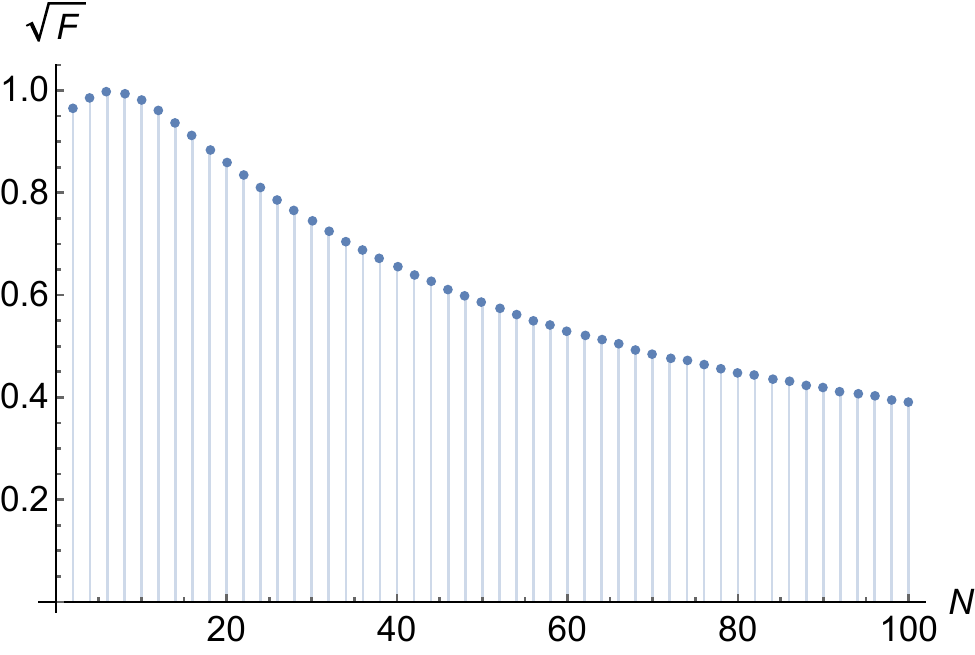}
    \caption{Left hand side: Overlap between states for optimal dPBT and optimal pPBT calculated for qubits by using expression~\eqref{eq:Fbetween} from Lemma~\ref{Fbetween2}. We see that for these two states the overlap between them saturates on the value 0.778. Right hand side: Overlap between states for non-optimal and optimal dPBT for qubits by using expression~\eqref{eq:l:FPBT} from Lemma~\ref{l:FPBT}. The maximal value of the overlap which is $F=0.9977$ is attained for $N=6$. In the asymptotic limit the both states are orthogonal. This plot has been firstly obtained in~\cite{deg} in the context of resource state degradation. In both figures we see completely different behaviour of the overlaps between the resource states.}
    \label{fig:my_label}
\end{figure}

\begin{figure}
    \centering
    \includegraphics[width=0.45\textwidth]{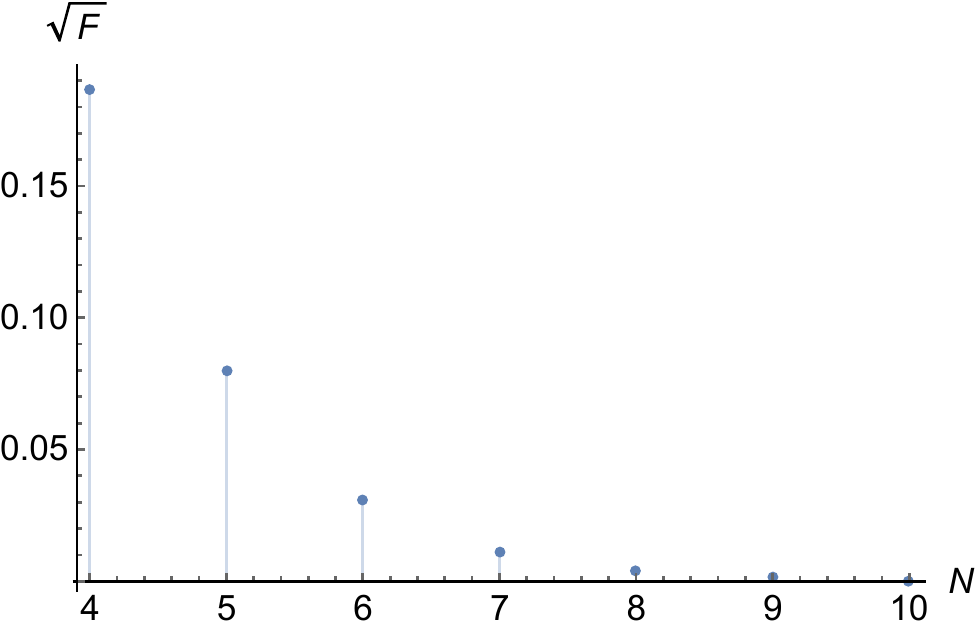}
    \caption{Overlap between states for non-optimal pPBT and optimal pPBT calculated for qubits by using expression~\eqref{72} from Lemma~\ref{ovpPBT}. We see that for these two states the overlap between them approaches 0 even for a very small number of ports $N$ making them orthonormal. This happens much faster than for the overlap between states for non-optimal and optimal dPBT depicted in Figure~\ref{fig:my_label}.}
    \label{fig:my_label2}
\end{figure}
For the reader convenience we collate results from the above lemmas in Table~\ref{tab_res_overlaps} for qubits and in Table~\ref{tab_res_overlaps2} for qudits.
\begin{table}[h!]
    \centering
    \begin{tabular}{c|cc}
      resource state  &  $|\Psi^+\>_{AB}$  &  $(O_A\otimes \mathbf{1}_B)|\Psi^+\>_{AB}$ \\ \hline 
\\[-1em]
       $|\Psi^+\>_{AB}$   &  1 & \large{$\sqrt{\frac{6}{2^N(N+1)(N+2)(N+3)}}\sum_{j=j_{\min}}^{N/2}\frac{(2j+1)^2}{\sqrt{\left(\frac{N}{2}-j\right)!\left(\frac{N}{2}+j+1\right)!}}$}\\ \hline 
\\[-1em]
       $(\widetilde{O}_A\otimes \mathbf{1}_B)|\Psi^+\>_{AB}$ & \large{$\sqrt{\frac{N!}{2^{N-2}(N+2)}}\sum_{j=j_{\min}}^{N/2}\frac{(2j+1)\operatorname{sin}\frac{\pi(2j+1)}{N+2}}{\sqrt{(\frac{N}{2}-j)!(\frac{N}{2}+j+1)!}}$} & \large{$\frac{2}{N+2}\sqrt{\frac{6}{(N+1)(N+3)}}\sum_{j=j_{\min}}^{N/2}(2j+1)\operatorname{sin}\left(\frac{\pi (2j+1)}{N+2}\right)$}\\
    \end{tabular}
    \caption{Results evaluated in Section~\ref{sec:fids} in the qubit case $d=2$. It contains expressions for the square-root fidelities $\sqrt{F}$ between the resource states in all variants of PBT protocols, for an arbitrary number of ports $N$. All the results are presented in spin angular momentum formalism. Operators $O_A,\widetilde{O}_A$ are optimising Alice's operations in optimal probabilistic and deterministic scheme respectively.}
    \label{tab_res_overlaps}
\end{table}

\begin{table}[h!]
    \centering
    \begin{tabular}{c|cc}
       resource state  &  $|\Psi^+\>_{AB}$  &  $(O_A\otimes \mathbf{1}_B)|\Psi^+\>_{AB}$ \\ \hline 
\\[-1em]
       $|\Psi^+\>_{AB}$   &  1 & \large{$\frac{1}{d^N}\sum_{\mu \vdash N}m_{\mu}\sqrt{g(N)d_{\mu}m_{\mu}}$}\\ \hline 
\\[-1em]
       $(\widetilde{O}_A\otimes \mathbf{1}_B)|\Psi^+\>_{AB}$ & \large{$\frac{1}{\sqrt{d^N}}\sum_{\mu \vdash N} v_{\mu}\sqrt{d_{\mu}m_{\mu}}$} & \large{$\sum_{\mu \vdash N}\frac{v_{\mu}m_{\mu}}{\sqrt{\sum_{\nu \vdash N}m_{\nu}^2}}$}\\
    \end{tabular}
    \caption{Results evaluated in Section~\ref{sec:fids}. It contains expressions for the square-root fidelities $\sqrt{F}$ between the resource states in all variants of PBT protocols, for an arbitrary number of ports $N$ and port dimension $d$. Operators $O_A,\widetilde{O}_A$ are optimising Alice's operations in optimal probabilistic and deterministic scheme respectively.}
    \label{tab_res_overlaps2}
\end{table}
\newpage
\section{Source code}
Here we provide the source code that generates all the requisite quantities. This Python code was executed on SageMath 9.0+ (Python 3) or later.

\begin{lstlisting}[language=Python]
def hook_formula(mu):
    return factorial(add(k for k in mu))/prod(mu.hook_length(i,j) for i,j in mu.cells())

import matplotlib.pyplot as plt
import numpy as np
from scipy.special import factorial
plt.style.use('seaborn-whitegrid')


d=2 #height
psucc_vals = []
fid_vals = []
#Ishizaka Hiroshima d=2, MES, popt
HI_MES_popt = []
#Ishizaka Hiroshima d=2, optimal state, popt
HI_opt_popt = []

#ratio of F of non-optimal dPBT and our p_succ
F_p_ratio = []
p_vals = []

for N in range (2,60):
#    print("item: ", N)
    p = Partitions(N-1,max_length=d) #generates all the shapes of height d
    p.list()

    #d_mu is the number of standard YTs and mu_alpha is the number of Semi-standard YTs, 
    #P_succ = sum_alpha sum_mu=(alpha + cell) d_mu*m_alpha

    res = 0

    
    for part in p:
        ssyt_alpha= SemistandardTableaux(part,max_entry=d) #semistandard tableaux that
        #correspond to a given partition
        m_alpha = ssyt_alpha.cardinality()
   
        #generate a list of partitions where we add a box to the existing one 
        part_plusone = list(part.up_list())
        #we want only the partitions which are obtained from part by adding 1 box 
        #that have height <=d
        part_plusone = [x for x in part_plusone if len(x)<=d] 

        for partplus in part_plusone:
            syt_mu = StandardTableaux(partplus)
            d_mu = syt_mu.cardinality()

            res +=d_mu*m_alpha         
            

    res = res/d**(N+1)
    psucc_vals.append(res)
    p_vals.append(1-(d^2-1)/(N+d^2-1))
    #p_vals.append(1-(N+2)/(2**(N+1)))
    
    res1 = 0;

    for part in p:
        part_res = 0;
        ssyt_alpha= SemistandardTableaux(part,max_entry=d) #semistandard tableaux that 
        #correspond to a given partition

        #generate a list of partitions where we add a box to the existing one 
        part_plusone = list(part.up_list())
        part_plusone = [x for x in part_plusone if len(x)<=d] #we want only partitions 
        #which are obtained from part by adding 1 box that have height <=d
        for partplus in part_plusone:
            syt_mu = StandardTableaux(partplus)
            ssyt_mu = SemistandardTableaux(partplus,max_entry=d)
            d_mu = syt_mu.cardinality()
            m_mu = ssyt_mu.cardinality()
            part_res += sqrt(d_mu*m_mu)
        res1 += part_res*part_res
    res1 = 1/d^(N+2)*res1
    fid_vals.append(res1)

    
    #Ishizaka Hiroshima d=2, MES, popt
    start =float(0.0)
    arr = []
    if (N-1) % 2 !=0: 
        start=0.5
    for s in np.arange (start, ((N-1)/2), 1):
        Num = (2*s+1)*(2*s+1)*factorial(N)
        Denom1 = (factorial((N-1)/2-s))
        Denom2 = (factorial((N+3)/2+s))
        arr.append(Num/(Denom1*Denom2))
    val = 1/2^(N)*sum(arr)

    HI_MES_popt.append(val)
        
    #Ishizaka Hiroshima d=2, optimal state, popt
    HI_opt_popt.append(1-3.0/(N+3))
    
    F_p_ratio.append(res1/res)
    
#line1 = plt.scatter(range(3,60),psucc_vals, label ='Psucc',linestyle='--')
#line2 = plt.scatter(range(3,60),p_vals, label ='Pvals',linestyle='dotted')
#line2, = plt.plot(fid_vals, label ="fidelity_opt",linewidth=2)
#line3, = plt.plot(HI_MES_popt, label='HI MES Psucc', linestyle='dotted')
#line4 = plt.scatter(range(3,20),HI_opt_popt, label='HI OPT Psucc', linestyle='dotted')

line5 = plt.scatter(range(2,60),F_p_ratio)
line6 = plt.scatter(range(2,60),fid_vals)

plt.legend((line5, line6),('F_det/p_mPBT', 'F_det'),numpoints=1, loc='lower right')

plt.xlabel("N")
plt.ylabel("Fidelity")
#plt.legend((line1, line2),('p_mPBT', 'p_OPT'),numpoints=1, loc='lower right')

\end{lstlisting}

\newpage
\bibliographystyle{unsrt}
\bibliography{biblio2}
\end{document}